\documentclass{amsart}

\usepackage{amssymb}
\usepackage{color}
\usepackage{array}
\usepackage{graphicx}

\newcommand{\rmd}{\mathrm{d}}
\newcommand{\rme}{\mathrm{e}}

\DeclareMathOperator{\Ren}{Re}
\DeclareMathOperator{\Imn}{Im}
\DeclareMathOperator{\modR}{mod}

\newcommand{\Jac}{\mathrm{Jac}}

\DeclareMathOperator{\Complex}{\mathbb{C}}
\DeclareMathOperator{\Integer}{\mathbb{Z}}
\DeclareMathOperator{\Natural}{\mathbb{N}}
\DeclareMathOperator{\wgt}{wgt}
\DeclareMathOperator{\ord}{ord}

\DeclareMathOperator{\res}{res}

\newcommand{\mFr}{\mathfrak{m}}

\newtheorem{theo}{Theorem}

\newtheorem{cor}{Corollary}

\theoremstyle{definition}
\newtheorem{rem}{Remark}

\theoremstyle{plain}

\title[Computation of $\wp$-functions ]{Computation of $\wp$-functions \\
on plane  algebraic curves}
\author{J Bernatska}
\address{}
\email{}
\date{\today}

\allowdisplaybreaks[4]

\begin{document}
 
\maketitle 
\begin{abstract}
Numerical tools for computation of $\wp$-functions, also known as Kleinian, or multiply periodic, are proposed.
In this connection, computation of periods of the both first and second kinds is reconsidered. 
An analytical approach to constructing  Riemann surfaces of plane algebraic curves of low gonalities  
is used.  The approach is based on explicit radical solutions to quadratic, cubic, and quartic equations,
which serve for hyperelliptic, trigonal, and tetragonal curves, respectively. 
The proposed analytical models of Riemann surfaces give full control over computation
of the Abel image of any point or divisor. 
Therefore,  computation of $\wp$-functions at Abel images of given divisors can be done directly.
An alternative computation with the help of the Jacobi inversion problem
is used for verification. Hyperelliptic and trigonal curves are considered in detail,
and illustrated by examples.
A method of finding the unique characteristic corresponding to the vector of Riemann constants
is suggested for non-hyperelliptic and hyperelliptic curves.
\end{abstract}

\textbf{Keywords:}  uniformization, sigma function, multiply periodic functions, 
Riemann surface model,  vector of Riemann constants

\textbf{MSC:}  32Q30, 30F10, 33F05, 65D20

\section{Introduction}
The interest to computing $\wp$-functions, also known as Kleinian after \cite{belHKF},
or multiply periodic after \cite{bakerMP}, arises from the realm of completely integrable systems, e.g.
the hierarchies of the Korteweg---de Vries equation (KdV), the sine-Gordon equation (SG), 
the non-linear Schr\"{o}dinger equation (NLS), etc., since finite-gap solutions can be expressed in terms of these functions, see \cite{belHKF,BerKdVRC2024}. Such a representation of solutions is rarely used, and convenient tools for 
computation are not developed.

Direct computation of $\wp_{1,1}$-function as a solution of the KdV equation 
is applied in \cite{BerKdVRC2024}. Since solutions are required to be real-valued, 
$\wp_{1,1}$-function is computed along a particular path in the Jacobian variety of a spectral curve. 

An alternative way of computing solutions of the mKdV, and  KdV equations
can be found in \cite{MatsDNA2022,MatsKdV2023}.
Though a solution in terms of $\wp_{1,1}$-function is used, 
it is expressed in terms of a divisor on a hyperelliptic curve
which solves  the Jacobi inversion problem. Thus, the required solution
is obtained  from computation of the inverse of the Abel map 
by means of Euler's  numerical quadrature. 
Computation of periods on a curve is avoided, 
as well as direct computation of $\wp_{1,1}$-function.

On the other hand, solutions of completely integrable equations in term of theta functions, 
see for example \cite{bbeim1994}, are widely used. 
And numerical tools for computation of the Riemann period matrix and the theta function
are well developed.

A powerful method of computing first kind period matrices (normalized and not normalized) associated with a 
given plane curve
 was presented in \cite{DH2001}, and implemented in the Maple package \texttt{algcurves}, see a detailed description in
 \cite{DP2011}. Also other packages for computing first kind period matrices, 
 the Abel map, and
 the theta function and its derivatives
 are created in Sage, Matlab, Julia, see \cite{AC2021} for a brief review.
 
Computation with the help of spectral approximation is suggested in \cite{FK2004}.
Linear combinations of Chebyshev polynomials are used in approximation of
integrands of first and third kind integrals between branch points, 
and integration is performed with the help of
the orthogonality relation on the polynomials.
Numerical simulation is performed in Matlab.
With Riemann period matrices computed by means of this technique,
solutions of the KdV and KP equations on hyperelliptic curves of genera $2$, $4$, $6$
are computed and illustrated in  \cite{FK2006}.
Solutions of the NLS equation and the Davey–Stewartson equation
on hyperelliptic curves of genera $2$ and $4$  are presented in \cite{KK2012}.
The spectral approximation technique allows to increase accuracy
 in almost solitonic cases, when pairs of branch points collide.

The known numerical tools are designed for studying theta-functional solutions.
On the contrary, computation of $\wp$-functions is lacking for 
appropriate  numerical tools.
First of all, periods of the both first and second kinds, subject to the Legendre relation,
are required. These periods are obtained from
first and second kind differentials which form an
associated\footnote{Fundamental integrals of the second kind
associated with the standard not normalized first kind integrals were introduced in \cite[\S\,138]{bakerAF}.}
system. Second kind periods are not covered by the known packages. 
Therefore, particular attention will be paid to computation of 
first and second kind periods on algebraic curves in question.

The problem of verifying results of computation gave rise to developing analytical models of 
Riemann surfaces of plane algebraic curves of low gonalities.
Such a model is constructed from explicit solutions obtained as radical expressions for roots of 
a quadratic, cubic, or quartic equation,
which serves for a hyperelliptic, trigonal, or tetragonal curve, respectively.
Continuous connection between these solutions is described  analytically in whole.
The obtained analytical model of a Riemann surface gives a full understanding how
to compute the Abel image of any point of the corresponding curve.
So, an effective analytical method of direct computation on plane
algebraic curves of low gonalities is proposed.

This analytical approach was initiated by V. Enolsky, and 
the hyperelliptic case with real branch points was developed by him. 
First and second kind periods obtained by this method served for 
verification of relations on theta functions, and $\wp$-functions.
However, the method had never been published before. In the present paper the method
is extended to hyperelliptic curves with arbitrary complex branch points, and to trigonal curves.
These two types of plane algebraic curves are the most demanded. 
In fact, only integrable systems with hyperelliptic spectral curves have been considered
in the literature as a matter of computation.

In addition to direct computation of $\wp$-functions,
computation based on the Jacobi inversion problem is also presented.
The latter is used for verification that obtained values of  $\wp$-functions are correct, and 
periods used for computing $\wp$-functions accommodate the curve in question. 
In the hyperelliptic case, generalizations of the Bolza formulas, which give expressions for branch points 
in terms of theta functions with characteristics, are used for verification of computed periods.

The paper is organized as follows. In Preliminaries the notion of Sato weight, the definitions
of $\sigma$-function and $\wp$-functions are given,  and also the Bolza formulas,
and solutions of the Jacobi inversion problem on hyperelliptic and trigonal curves.
In section~\ref{s:HypPer}  an analytical model of the Riemann surface of a hyperelliptic curve,
and computation of periods are explained in detail. Section~\ref{s:HypWP} 
shows computation of the Abel images of arbitrary points of a hyperelliptic curve, and
$\wp$-functions on non-special divisors. In section~\ref{s:TrigPer}  
an analytical model of the Riemann surface of a trigonal curve and 
 computation of periods are presented. Section~\ref{s:TrigWP} 
illustrates the trigonal case with  computations of the Abel images of arbitrary points, and 
$\wp$-functions.

The proposed method is illustrated by examples: hyperelliptic curves of genus~$4$ with (1)
complex branch points, and (2) all real branch points, and a trigonal curve of genus~$3$.
Computations are made in Wolfram Mathematica 12. The full code is posted in the Wolfram Community portal, see \\
\texttt{https://community.wolfram.com/groups/-/m/t/3243472}\\
\texttt{https://community.wolfram.com/groups/-/m/t/3252458}

\section{Preliminaries}
\subsection{Sato weight}
The notion of \emph{Sato weight} plays an important role in the theory of $(n,s)$-curves. Such a curve arises  as
a universal unfolding of the Pham singularity $-y^n + x^s = 0$ with co-prime $n$ and $s$, $n<s$.
Thus, an $(n,s)$-curve $\mathcal{C}$ is defined by 
\begin{subequations}\label{nsCurve}
\begin{align}
\mathcal{C} &=\{(x,y)\in \Complex^2 \mid f(x,y) =0\},\\
&\quad  f(x,y) \equiv -y^n + x^s + \sum_{j=0}^{n-2} \sum_{i=0}^{s-2}  \lambda_{ns-in- js} y^j x^i, \label{fEq}\\
&\quad \lambda_{k\leqslant 0}=0, \quad \lambda_{k}\in \Complex. \label{ModCond}
\end{align}
\end{subequations}
where $\lambda_{k}$ serve as parameters of the curve, and $k$ shows the Sato weight of $\lambda_k$.
Only parameters with positive  weights are allowed. 
The Sato weights of $x$ and $y$ are $\wgt x = n$, and $\wgt y = s$.  Then $\wgt f(x,y) = n s$.
Note, that some terms are omitted in \eqref{nsCurve}, since the definition contains the minimal
number of parameters. All extra terms can be eliminated by a proper bi-rational transformation.

Due to $n$ and $s$ are co-prime, infinity is a Weierstrass point, and 
a branch point  where all $n$ sheets of the curve wind.
Let  $\xi$ denote a local parameter near infinity, then 
\begin{equation}\label{param}
x=\xi^{-n},\qquad y = \xi^{-s}(1+O(\lambda))
\end{equation}
gives the simplest parametrization of \eqref{nsCurve}. Evidently, the
Sato weight equals the negative exponent of the leading term in the expansion near infinity.

\subsection{Abel map}
Let  $\rmd u = (\rmd u_{\mathfrak{w}_1} $, $\rmd u_{\mathfrak{w}_2}$, $\dots$, $\rmd u_{\mathfrak{w}_g} )^t$
be not normalized first kind differentials, labeled by elements of the Weierstrass gap sequence 
$\mathfrak{W} = \{\mathfrak{w}_1$,
$\mathfrak{w}_2$, \ldots, $\mathfrak{w}_g\}$, 
which coincide with negative weights of the differentials: 
$\wgt \rmd u_{\mathfrak{w}_i} = - \mathfrak{w}_i$,
and show the orders of zero at infinity.

Let the Abel map $\mathcal{A}$ be constructed with  not normalized differentials $\rmd u$:
\begin{gather}\label{AbelM}
 \mathcal{A}(P) = \int_{\infty}^P \rmd u,\qquad P=(x,y)\in \mathcal{C}.
\end{gather}
Here infinity is used as the base-point, which is the standard choice in the case of $(n,s)$-curves.

First kind integrals along canonical homology cycles $\{\mathfrak{a}_i,\,\mathfrak{b}_i\}_{i=1}^g$ give
first kind period matrices:
\begin{gather}\label{omegaM}
 \omega = (\omega_{ij})= \bigg( \int_{\mathfrak{a}_j} \rmd u_i\bigg),\qquad\quad
 \omega' = (\omega'_{ij}) = \bigg(\int_{\mathfrak{b}_j} \rmd u_i \bigg).
\end{gather}
Columns of $\omega$, $\omega'$ generate the lattice $\{\omega, \omega'\}$ of periods.
Then $\Jac(\mathcal{C})=\Complex^g \backslash \{\omega, \omega'\}$ is the 
Jacobian variety  of the curve $\mathcal{C}$, equipped with not normalized coordinates
$u = (u_{\mathfrak{w}_1},u_{\mathfrak{w}_2}, \dots, u_{\mathfrak{w}_g})^t$.
The coordinates are labeled by elements of the Weierstrass gap sequence,
and $\wgt  u_{\mathfrak{w}_i} = - \mathfrak{w}_i$.

Let $v = \omega^{-1} u$ be normalized coordinates on the Jacobian variety, and
$(1_g,\tau)$ be normalized periods, where $1_g$ denotes the identity matrix of order $g$, 
and $\tau = \omega^{-1}\omega'$. Matrix $\tau$ is symmetric with a positive imaginary part: 
$\tau^t=\tau$, $\Imn \tau >0$,
that is $\tau$ belongs to the Siegel upper half-space. The Sato weight is not associated with 
normalized coordinates.
The normalised first kind differentials are defined  by
\begin{gather*}
 \rmd v = \omega^{-1} \rmd u,
\end{gather*}
and the Abel map $\bar{\mathcal{A}}$ with respect to the normalized differentials is
\begin{gather}\label{AbelMNorm}
 \bar{\mathcal{A}}(P) = \int_{\infty}^P \rmd v,\qquad P=(x,y)\in \mathcal{C}.
\end{gather}

\subsection{Theta and sigma functions}
Recall the two entire functions on $\Complex^g \supset \Jac(\mathcal{C})$,
which generate multiply periodic (or abelian) functions, and so serve 
for uniformization of a curve  $\mathcal{C}$.

The Riemann \emph{theta function} (or $\theta$-function) 
\begin{gather}\label{ThetaDef}
 \theta(v;\tau) = \sum_{n\in \Integer^g} \exp \big(\imath \pi n^t \tau n + 2\imath \pi n^t v\big)
\end{gather}
is defined in terms of normalized coordinates
 $v$, and normalized period matrix $\tau$.
Theta function with characteristic is defined by
\begin{equation}\label{ThetaDefChar}
 \theta[\varepsilon](v;\tau) = \exp\big(\imath \pi (\tfrac{1}{2} \varepsilon'{}^t) \tau (\tfrac{1}{2}\varepsilon')
 + 2\imath \pi  (v+\tfrac{1}{2}\varepsilon)^t (\tfrac{1}{2}\varepsilon')\big)  \theta(v+\tfrac{1}{2} \varepsilon + \tau ( \tfrac{1}{2}\varepsilon');\tau),
\end{equation}
where a characteristic is a $2\times g$ matrix $[\varepsilon]= (\varepsilon', \varepsilon)^t$
with real values within the interval $[0,2)$.  Every point $u$ in the fundamental domain of $\Jac(\mathcal{C})$ 
can be represented by its characteristic $[\varepsilon]$, namely
\begin{equation*}
u =  \tfrac{1}{2} \omega \varepsilon +   \tfrac{1}{2} \omega' \varepsilon'.
\end{equation*}

In the hyperelliptic case, the Abel images of branch points, and any combination of branch points
are described by characteristics with components $1$ or $0$, 
which are called  half-integer characteristics. Such a characteristic
is odd whenever $\varepsilon^t \varepsilon'  = 0$ ($\modR 2$), 
and even whenever $\varepsilon^t \varepsilon' = 1$ ($\modR 2$). $\theta$-Function with 
half-integer characteristic
has the same parity as its characteristic.

The entire function on $\Complex^g \supset \Jac(\mathcal{C})$ covariant under integer shifts 
on the period lattice is called the \emph{sigma function} (or $\sigma$-function), see \cite[p.\,97]{bakerMP}.
As a definition we use its relation with  $\theta$-function, see \cite[Eq.(2.3)]{belHKF}:
\begin{equation}\label{SigmaThetaRel}
\sigma(u) = C \exp\big({-}\tfrac{1}{2} u^t \varkappa u\big) \theta[K](\omega^{-1} u;  \omega^{-1} \omega'),
\end{equation}
where $[K]$ denotes the characteristic of the vector of Riemann constants, and 
a symmetric matrix $\varkappa = \eta \omega^{-1}$ is  obtained from the second kind period matrix $\eta$. 
The Sato weight of  $\sigma$-function is 
 $\wgt \sigma = - (n^2-1)(s^2-1)/24$, see \cite{bel99}.

$\sigma$-Function is defined in terms of not normalized coordinates $u$, and  
not normalized period matrices of the first kind $\omega$, $\omega'$, and the second kind $\eta$, $\eta'$.
The latter are defined as follows
\begin{gather}\label{etaM}
 \eta = (\eta_{ij})= \bigg( \int_{\mathfrak{a}_j} \rmd r_i\bigg),\qquad\quad
 \eta' = (\eta'_{ij}) = \bigg(\int_{\mathfrak{b}_j} \rmd r_i \bigg),
\end{gather}
with second kind differentials $\rmd r = (\rmd r_{\mathfrak{w}_1} $, $\rmd r_{\mathfrak{w}_2}$, $\dots$, $\rmd r_{\mathfrak{w}_g} )^t$. It is important to choose the second kind differentials which form an associated system with
differentials of the first kind, see \cite[\S\,138]{bakerAF}. Note that $\rmd r_{\mathfrak{w}_i}$
has the only pole of order $\mathfrak{w}_i$ at infinity, and so $\wgt r_{\mathfrak{w}_i} = \mathfrak{w}_i$.
In the vicinity of infinity, with a local parameter $\xi$ such that $\xi(\infty)=0$, the following relation holds 
\begin{equation}\label{urRel}
\res_{\xi=0} \Big(\int_0^\xi \rmd u(\tilde{\xi}) \Big) \rmd r(\xi)^t = 1_g,
\end{equation}
which completely determines the principle part of $\rmd r(\xi)$. 

The not normalized period matrices of the first $\omega$, $\omega'$ and second
$\eta$, $\eta'$ kinds satisfy the Legendre relation, see \cite[\S\,140]{bakerAF},
\begin{gather}\label{LegRel}
\Omega^t \mathrm{J}\, \Omega = 2\pi \imath \mathrm{J},\\
\Omega = \begin{pmatrix} \omega & \omega' \\
\eta & \eta' \end{pmatrix},\qquad 
\mathrm{J} = \begin{pmatrix} 0 & - 1_g \\ 1_g & 0 \end{pmatrix}. \notag
\end{gather}

Multiply periodic $\wp$-functions are defined with the help of $\sigma$-function:
\begin{gather*}
\wp_{i,j}(u) = -\frac{\partial^2 \log \sigma(u) }{\partial u_i \partial u_j },\qquad
\wp_{i,j,k}(u) = -\frac{\partial^3 \log \sigma(u) }{\partial u_i \partial u_j \partial u_k},\quad \text{etc.}
\end{gather*}
Since  $\sigma$-function vanishes on special divisors according to the Riemann vanishing theorem,
$\wp$-functions are defined on $\Jac(\mathcal{C}) \backslash \Sigma$, where $\Sigma = \{u \mid \sigma(u)=0\}$.

From \eqref{SigmaThetaRel} we obtain expressions for $\wp$-functions in terms of $\theta$-function:
\begin{gather}\label{WPdef}
\begin{split}
&\wp_{i,j}(u)  = \varkappa_{i,j} - \frac{\partial^2}{\partial u_i \partial u_j } \log \theta[K](\omega^{-1} u; \omega^{-1} \omega'),\\
& \wp_{i,j,k}(u)  =- \frac{\partial^3}{\partial u_i \partial u_j \partial u_k} \log \theta[K](\omega^{-1} u; \omega^{-1} \omega').
\end{split}
\end{gather}

The vector of Riemann constants $K$ with respect to a base-point $P_0$ is defined by the formula,
\cite[Eq.\,(2.4.14)]{Dub1981}
\begin{gather}\label{VRC}
K_j = \frac{1}{2}(1+\tau_{j,j}) - \sum_{l\neq j} \oint_{\mathfrak{a}_l} \rmd v_j (P)
\int_{P_0}^P \rmd v_l,\quad j=1,\,\dots,\, g.
\end{gather}
In the hyperelliptic case, the vector is computed\footnote{Below, we use a homology basis different from 
\cite{fay973}, and so $[K]$ is not exactly the same, but computed in the same way.} in \cite[p.\,14]{fay973}, and  equals 
the sum of all odd characteristics of the fundamental set of $2g+1$ characteristics
which represent  branch points, according to \cite[\S\,200--202]{bakerAF}.

\subsection{Bolza formulas and generalizations}
In genus $2$, expressions for branch points in terms of $\theta$-function are known as the
Bolza formulas
\begin{gather*}
 e_\iota = -\frac{\partial_{u_3} \theta[\{\iota\}] (\omega^{-1} u)}
 {\partial_{u_1} \theta[\{\iota\}] (\omega^{-1} u)} \Big|_{u=0},
\end{gather*}
where $[\{\iota\}]$ denotes the characteristic corresponding to a branch point $e_\iota$,
see \cite[Eq.\,(6)]{Bolza}.
A generalization of the Bolza formulas 
for a hyperelliptic curve of arbitrary genus $g$ is obtained in \cite{BerTF2020}.
In particular,
\begin{gather*}
 e_\iota = - \frac{\partial^{[g/2]}_{u_{2(g\modR 2)+1},\dots,u_{2g-7} ,u_{2g-1}} \theta[\{\iota\}](\omega^{-1} u)}
 {\partial^{[g/2]}_{u_{2(g\modR 2)+1},\dots,u_{2g-7} ,u_{2g-3}} \theta[\{\iota\}](\omega^{-1} u)}\Big|_{u=0}.
\end{gather*}

\subsection{Jacobi inversion problem}
Given a point $u$ of the Jacobian variety  $\Jac(\mathcal{C})$
find a reduced divisor $D \in \mathcal{C}^g$ such that $\mathcal{A}(D) = u$.
Every class of linearly equivalent divisors on a curve of genus $g$
 has a representative called a reduced divisor, let it be a positive divisor of degree  $g$ or less. 
Reduced divisors of degree less than $g$ are special, and $\theta[K]$ vanishes 
on such divisors, according to the Riemann vanishing theorem. 
Reduced divisors of degree $g$ are  non-special.
Every non-special divisor represents its class uniquely. 
The Abel image of the subspace of non-special divisors
coincides with $\Jac(\mathcal{C}) \backslash \Sigma$. 

A solution of the Jacobi inversion problem is known for non-special divisors.
On hyperelliptic curves such a solution was given in  \cite[\S\;216]{bakerAF}
and rediscovered in \cite[Theorem 2.2]{belHKF}.
Let a non-degenerate hyperelliptic curve of genus $g$ be defined\footnote{A $(2,2g+1)$-curve serves as a 
canonical form of hyperelliptic curves of genus $g$.} by
\begin{equation}\label{V22g1Eq}
-y^2 + x^{2 g+1} + \sum_{i=1}^{2g} \lambda_{2i+2} x^{2g-i} = 0.
\end{equation}
Let $u = \mathcal{A}(D)$ be the Abel image of  a degree $g$ positive non-special  divisor  $D$
on the curve. Then $D$ is uniquely defined by the system of equations 
\begin{subequations}\label{EnC22g1}
\begin{align}
&\mathcal{R}_{2g}(x;u) \equiv x^{g} -  \sum_{i=1}^{g} x^{g-i}  \wp_{1,2i-1}(u) = 0,\\ 
&\mathcal{R}_{2g+1}(x,y;u) \equiv 2 y + \sum_{i=1}^{g} x^{g-i}  \wp_{1,1,2i-1}(u) = 0.
\end{align}
\end{subequations}

On a trigonal curve, the Jacobi inversion problem is solved in \cite{bel00}.
A method of obtaining such a solution on a curve of an arbitrary gonality
is presented in \cite{BLJIP22};  trigonal, tetragonal and pentagonal curves are considered as an illustration.
In the case of a $(3,3\mFr +1)$-curve, a degree $g$ positive non-special  divisor  $D$
such that $u = \mathcal{A}(D)$ is given by the system
\begin{subequations}\label{EnC33m1}
\begin{align}
&\mathcal{R}_{6\mFr}(x,y;u) \equiv x^{2\mFr} 
-  y \sum_{i=1}^{\mFr} \wp_{1,3i-2}(u) x^{\mFr-i}
-  \sum_{i=1}^{2\mFr} \wp_{1,3i-1}(u) x^{2\mFr-i} = 0,\\ 
&\mathcal{R}_{6\mFr+1}(x,y;u) \equiv 2 y x^{\mFr} 
+ y \sum_{i=1}^{\mFr} \big(\wp_{1,1,3i-2}(u) - \wp_{2,3i-2}(u)  \big) x^{\mFr-i} \\
&\qquad\qquad\qquad\qquad\qquad + \sum_{i=1}^{2\mFr} \big(\wp_{1,1,3i-1}(u) - \wp_{2,3i-1}(u) \big) x^{2\mFr-i}= 0. \notag
\end{align}
\end{subequations}
In the case of a $(3,3\mFr +2)$-curve, by the system
\begin{subequations}\label{EnC33m2}
\begin{align}
&\mathcal{R}_{6\mFr+2}(x,y;u) \equiv y x^{\mFr} 
-  y \sum_{i=1}^{\mFr} \wp_{1,3i-1}(u) x^{\mFr-i} 
-  \sum_{i=1}^{2\mFr+1} \wp_{1,3i-2}(u) x^{2\mFr+1-i} = 0,\\ 
&\mathcal{R}_{6\mFr+3}(x,y;u) \equiv 2 x^{2\mFr+1} 
+ y \sum_{i=1}^{\mFr} \big( \wp_{1,1,3i-1}(u) - \wp_{2,3i-1}(u) \big) x^{\mFr-i} \\
&\qquad\qquad\qquad\qquad\qquad - \sum_{i=1}^{2\mFr+1} \big(\wp_{1,1,3i-2}(u) - \wp_{2,3i-2}(u) \big) x^{2\mFr+1-i}= 0. \notag
\end{align}
\end{subequations}

\section{Periods on a  hyperelliptic curve}\label{s:HypPer}

Hyperelliptic curves are the best  known plane algebraic curves. There exists a universal approach to
choosing a homology and cohomology bases on such curves, 
as well as constructing the corresponding Riemann surfaces.

\subsection{Hyperelliptic curves}
Let a generic  hyperelliptic curve  be defined by 
\begin{equation}\label{HypC}
0 = \tilde{f}(x,y) \equiv - y^2 + y \mathcal{Q}(x) + \mathcal{P}(x),
\end{equation}
where $\deg \mathcal{P} = 2g+1$  or  $2g+2$, in the case of genus $g$.

An equation with $\deg \mathcal{P}=2g+1$, $\mathcal{Q}(x) \equiv 0$ defines an $(n,s)$-curve, 
which is considered as the canonical form of a hyperelliptic curve of genus $g$:
\begin{equation}\label{HypCCanon}
 0 = f(x,y) \equiv - y^2 + x^{2g+1} + \lambda_4 x^{2g-1} + \cdots + \lambda_{4g} x + \lambda_{4g+2}.
\end{equation}
The Sato weights are $\wgt x = 2$, $\wgt y = 2g+1$, and so $\wgt f = 4g+2$. 

The term $y \mathcal{Q}(x)$ in \eqref{HypC} 
 is eliminated by the map $y \mapsto \tilde{y} +\tfrac{1}{2} \mathcal{Q}(x)$, which leads to
\begin{equation*}
0 = \tilde{f}(x,\tilde{y})\equiv  - \tilde{y}^2 +  \widetilde{\mathcal{P}}(x),\qquad 
\widetilde{\mathcal{P}}(x)  = \mathcal{P}(x)  + \tfrac{1}{4} \mathcal{Q}(x)^2.
\end{equation*}
Thus, the discriminant of \eqref{HypC} is defined by the formula
\begin{equation}
\Delta(x) = \widetilde{\mathcal{P}}(x) = \mathcal{P}(x) + \tfrac{1}{4} \mathcal{Q}(x)^2.
\end{equation}

A canonical curve can be defined by its branch points $\{(e_i,\,0)\}_{i=1}^{2g+1}$, namely
\begin{gather}\label{HypCe}
0 = f(x,y)= - y^2 + \mathcal{P}(x),\qquad \quad \mathcal{P}(x) = \prod_{j =1}^{2g+1}  (x-e_j).
\end{gather}
For the sake of brevity, the notation $e_i$ is employed both for a branch point $(e_i,\,0)$ 
and its $x$-coordinate, in the hyperelliptic case. 
If all branch points are distinct, then the curve is non-degenerate, and its genus equals $g$. 
The curve \eqref{HypCe} has also a branch point located at infinity,  referred  as $e_0$.
Finite branch points $\{e_i\}$ of the  curve \eqref{HypCCanon} satisfy the condition
$\sum_{i=1}^{2g+1} e_i = 0$. Without  this condition, $\mathcal{P}$ in \eqref{HypCe}
contains also the term $\lambda_2 x^{2g}$.
In what follows, we omit such a condition, and allow $\{e_i\}$ be arbitrary.

A curve with $\deg \mathcal{P}=2g+2$, and $\mathcal{Q}(x) \equiv 0$ has $2g+2$ finite branch points
 $\{(e_i,\,0)\}_{i=0}^{2g+1}$. The corresponding canonical form is obtained by
moving the finite branch point $e_0$ to infinity by means of a proper M\"{o}bius transformation.

In the generic case \eqref{HypC}, the maximal $\deg \mathcal{Q}$ equals $g$, which respects the Sato weight,
and guarantees that the genus of the curve does not exceed $g$.
Such a curve can be defined by choosing arbitrary values $\{e_i\}$ of number $2g+1$ or $2g+2$,
and choosing a polynomial $\mathcal{Q}$ of a degree up to $g$. Then the corresponding $y$-coordinates of 
branch points $\{B_i=(e_i,h_i)\}$  are obtained by the formula $h_i = \tfrac{1}{2} \mathcal{Q}(e_i)$.

The Weierstrass gap sequence of a hyperelliptic curve \eqref{HypC} is
\begin{gather*}
\mathfrak{W} = \{\mathfrak{w}_i = 2i-1 \mid i=1,\dots, g\}.
\end{gather*}

In what follows, we focus on the  form \eqref{HypCe} of a hyperelliptic curve.

\subsection{Riemann surface}\label{ss:RiemSurf}
At every point $x$, except branch points,
there exist two values of $y$ ($\textsf{s} = \pm 1$):
\begin{subequations}\label{yDefH}
\begin{align}
&y_{\textsf{s}}(x) = \textsf{s} \sqrt{\Delta(x)}, \quad \Delta(x) =  \mathcal{P}(x),& 
&\text{in the  canonical case,} 
\qquad \text{or}  \label{yDefHcanon} \\
&y_{\textsf{s}}(x) = \tfrac{1}{2} \mathcal{Q}(x) + \textsf{s} \sqrt{\Delta(x)},&
&\text{in the  generic case.}
\end{align} 
\end{subequations}

Let the square root function be defined as follows 
\begin{gather}\label{SqrtDef}
\sqrt{\Delta(x)} = \left\{ \begin{array}{ll}
\sqrt{|\Delta(x)|} \,\rme^{(\imath/2) \arg  \Delta(x)} 
& \text{ if }\quad \arg  \Delta(x) \geqslant 0, \\
\sqrt{|\Delta(x)|} \, \rme^{(\imath/2) \arg \Delta(x)+ \imath \pi} 
& \text{ if }\quad \arg  \Delta(x) < 0,
\end{array}  \right.
\end{gather}
where $\arg$ has the range $(-\pi,\, \pi]$.
With  such a definition the range of $\arg \sqrt{\Delta(x)}$ is $[0,\pi)$.
Moreover, 
\begin{theo}\label{T:DeltaDscont}
Let $\sqrt{\Delta}$ be defined by \eqref{SqrtDef}. Then
$y_{\textsf{s}}$ defined by \eqref{yDefH} 
 have  discontinuity over the contour  $\Gamma = \{x \mid \arg \Delta(x) = 0\}$,
 and $y_{+}$ serves as the analytic continuation of $y_{-}$ on the other side of the contour, 
 and vice versa. 
\end{theo}
\begin{proof}
Let $\tilde{x}$ be located in the vicinity of the contour $\Gamma$, more precisely 
$|\arg \Delta(\tilde{x})| < 2 \phi$, with a small positive value $\phi$. 
Then $0 \leqslant \arg \sqrt{\Delta(\tilde{x})} < \phi$  if $\arg \Delta(\tilde{x}) \geqslant 0$,
and $\pi - \phi < \arg \sqrt{\Delta(\tilde{x})} < \pi$ if $\Delta(\tilde{x}) < 0$.
Evidently, the discontinuity of $\sqrt{\Delta}$ is located over the contour $\Gamma$.

Next, we find the analytic continuation of $\sqrt{\Delta}$.
Let $U(x_0; \delta)$ be a disc of radius~$\delta$ with the center at $x_0 \in \Gamma$.
The contour $\Gamma$ divides the disc into two parts: $U_+$ where $\arg \Delta(x) \geqslant 0$,  and  
$U_-$ where $\arg \Delta(x) < 0$. Let $2 \phi_+ \,{=}\, \max_{\tilde{x} \in U_+} \arg \Delta(\tilde{x})$, 
then the range of $\arg \sqrt{\Delta}$ on $U_+$ is $[0,\phi_+)$.
Let ${-}2 \phi_- \,{=}\, \min_{\tilde{x} \in U_-} \arg \Delta(\tilde{x})$,
then the range of $\arg \sqrt{\Delta}$ on $U_-$ is $(\pi - \phi_-,\pi)$.
Thus, the analytic continuation of $\sqrt{\Delta}$ from $U_+$ to $U_-$ 
is  $-\sqrt{\Delta}$, since the range of 
$\arg (-\sqrt{\Delta(x)}) = \arg (\rme^{-\imath \pi} \sqrt{\Delta(x)})$  on $U_-$  is $(- \phi_-,0)$. 
And the analytic continuation of $\sqrt{\Delta}$
from $U_-$ to  $U_+$ is given by $-\sqrt{\Delta} = \rme^{\imath \pi} \sqrt{\Delta}$,
with the range $[\pi,  \pi + \phi_+)$ of $\arg (-\sqrt{\Delta}) $ on $U_+$.
\end{proof}

With the help of definition \eqref{SqrtDef}, we fix the position of 
discontinuity of $\sqrt{\Delta}$ at the contour $\Gamma$. 
This allows to determine connection of solutions $y_{\textsf{s}}$
on the Riemann surface, and mark sheets.

In order to construct the Riemann surface of a curve,
we choose a continuous path $\gamma$ on the Riemann sphere through all $\{e_i\}$,
starting and ending at $e_0 = \infty$.
Along $\gamma$, at any intersection with $\Gamma$ 
the sign $\textsf{s}$ in \eqref{yDefH} changes into the opposite.
Thus, sequences of signs along  $\gamma$ lifted to each sheet
mark the sheets. And so  the \emph{monodromy}
of the Riemann surface is defined, cf.\;\cite[Sect.\,2.7]{DP2011}.
We call $\gamma$ the \emph{monodromy path}.

\subsection{Monodromy path}
The monodromy path $\gamma$ is
the \emph{key element} of the proposed scheme of  constructing an analytical model of the Riemann surface
of $\mathcal{C}$. 

The  path starts at infinity as $x\to -\infty$, goes through all branch points in a chosen order, 
and ends at infinity as $x\to \infty$.
Below, an algorithm of drawing such a path and marking sheets is presented. 
\begin{enumerate}
\renewcommand{\labelenumi}{\arabic{enumi}.}
\item Let all finite branch points $\{e_{i}\}_{i=1}^{2g+1}$ be sorted 
ascendingly first by the real part, then by the imaginary part.
\item According to this order a continuous path $\gamma$ through all $e_i$
is constructed from straight line segments $[e_i,e_{i+1}]$, $i=1$, \ldots, $2g$. 
Then the segment $(-\infty, e_1]$ is added at the beginning
of the polygonal path, and $[e_{2g+1},\infty)$ at the end.
The path goes below the points $\{e_i\}$, and so 
 in the counter-clockwise direction near each $e_i$.
Such a path is marked in orange  on fig.\,\ref{f:BPCuts}.

\item Plot the contour $\Gamma = \{x \mid \arg \mathcal{P}(x) = 0\}$, which
consists of segments $\Gamma_i$ between $e_i$ and infinity,
see blue contours on fig.\,\ref{f:BPCuts}.
Find the sequence of signs $\{\textsf{s}_{0,1}\}\cup \{\textsf{s}_{i,i+1}\}_{i=1}^{2g} \cup \{\textsf{s}_{2g+1,0}\}$
corresponding to the sequence of segments of the path $\gamma$,
starting with $\textsf{s}_{0,1}=+1$. Index $0$ is used for infinity. At any intersection with $\Gamma$ 
the sign changes into the opposite. The sequence of signs determines connection between solutions $y_{\textsf{s}}$
on Sheet\;\textsf{a}. Sheet\;\textsf{b} is marked by the sequence with the opposite sign on each segment. 
On a hyperelliptic curve, Sheet\;\textsf{a} is sufficient for all computations.
\end{enumerate}

\begin{rem}
It could happen, that an intersection of a polygonal path~$\gamma$  with  $\Gamma$
is caused by curling of $\Gamma$  around a branch point. 
Such an intersection could be avoided by continuous deformation of  $\gamma$.
\end{rem}

\subsection{Homology}\label{ss:Hom}
Cuts are made between points $e_{2k-1}$ and $e_{2k}$ with $k$ from $1$ to $g$, and from $e_{2g+1}$ to infinity. 
With $k$ running from $1$ to $g$ 
an $\mathfrak{a}_k$-cycle encircles the cut $(e_{2k-1},\, e_{2k})$ counter-clockwise, and a $\mathfrak{b}_k$-cycle 
enters  the cut $(e_{2k-1},\, e_{2k})$ and emerges from the cut $(e_{2g+1},\,\infty)$, 
see fig.\,\ref{f:PathCycles} as an example.
This canonical homology basis is adopted from Baker \cite[p.\,297]{bakerAF},
and can be considered as standard on a hyperelliptic curve.

\begin{rem}
In the case of $\deg \mathcal{P}=2g+2$,  with  branch points
$\{(e_i,0)\}_{i=0}^{2g+1}$, we  sort the latter in the same way as in the canonical case.
So $e_0$ has the smallest real and imaginary parts among all branch points.
Cuts are made between points $e_{2k-1}$ and $e_{2k}$ with $k$ from $1$ to $g$, 
and from $e_{2g+1}$ to $e_0$ through infinity. A canonical homology basis is introduced in
a similar way: an $\mathfrak{a}_k$-cycle encircles the cut $(e_{2k-1},\, e_{2k})$ counter-clockwise, and a 
$\mathfrak{b}_k$-cycle enters   the cut $(e_{2k-1},\, e_{2k})$ 
and emerges from the cut $(e_{2g+1}, \infty) \cup (\infty, e_{0})$.
\end{rem}

\subsection{Cohomology}\label{ss:coHom}
First kind differentials are defined in the standard way, see \cite[\textit{Ex.\,i}, p.\,195]{bakerAF} for example,
\begin{align}\label{HDifCg}
\rmd u_{2i-1} = \frac{x^{g-i} \rmd x}{\partial_y f(x,y)},\qquad i=1,\,\dots,\,g.  
\end{align} 
On the canonical hyperelliptic curve, the second kind differentials  associated  with the first kind ones
have the form ($\lambda_0=1$)
\begin{align}\label{coHDifCg}
&\rmd r_{2i-1} = \frac{\rmd x}{\partial_y f(x,y)} \sum_{k=1}^{2i-1} k \lambda_{4i-2k-2} x^{g-i+k},
\qquad i=1,\,\dots,\,g. 
\end{align} 

\subsection{Computation of periods}
First kind integrals on each segment along the polygonal monodromy path $\gamma$ 
lifted to Sheet\;\textsf{a} are  computed by
\begin{subequations}\label{Aint}
\begin{align}
&\mathcal{A}_{i,i+1}^{[\textsf{s}_{i,i+1}]} = \int_{e_i}^{e_{i+1}} \rmd u^{[\textsf{s}_{i,i+1}]} ,\qquad i=1,\, \dots,\, 2g,\\
&\mathcal{A}_{0,1}^{[\textsf{s}_{0,1}]}  = \int_{-\infty}^{e_{1}} \rmd u^{[\textsf{s}_{0,1}]},\qquad \qquad
\mathcal{A}_{2g+1,0}^{[\textsf{s}_{2g+1,0}]} = \int_{e_{2g+1}}^{\infty} \rmd u^{[\textsf{s}_{2g+1,0}]}.
\end{align}
\end{subequations}
The integrand of $\mathcal{A}_{i,j}^{[\textsf{s}_{i,j}]}$ is taken with the  sign $\textsf{s}_{i,j}$, that is
$$\rmd u^{[\textsf{s}_{i,j}]} = \begin{pmatrix} x^{g-1} \\ \vdots \\ x \\ 1 \end{pmatrix} 
\frac{\rmd x}{-2\, \textsf{s}_{i,j} \sqrt{\mathcal{P}(x)}}.$$
Due to the involution of a hyperelliptic curve,
the following relations hold
\begin{gather}\label{HyperInvRels}
\sum_{k=1}^{g} \mathcal{A}_{2k-1, 2k}^{[\textsf{s}_{2k-1,2k}]}   + \mathcal{A}_{2g+1,0}^{[\textsf{s}_{2g+1,0}]}  = 0,\qquad
\mathcal{A}_{0,1}^{[\textsf{s}_{0,1}]}  + \sum_{k=1}^{g} \mathcal{A}_{2k, 2k+1}^{[\textsf{s}_{2k,2k+1}]}    = 0,
\end{gather}
which serve for verification.

 According to the choice of canonical cycles, columns of the first kind period matrices are  
\begin{gather}\label{FKPerN}
\omega_k = 2\mathcal{A}_{2k-1, 2k}^{[\textsf{s}_{2k-1,2k}]} ,\qquad\qquad  
\omega'_k = - 2\sum_{j=k}^g \mathcal{A}_{2j, 2j+1}^{[\textsf{s}_{2j,2j+1}]} .
\end{gather}
The normalized period matrix, which is the Riemann period matrix, is obtain by 
$$\tau = \omega^{-1} \omega',$$
and required  to be symmetric with a positive imaginary part.

Second kind integrals $\mathcal{B}_{i,j}^{[\textsf{s}_{i,j}]}$ are computed similarly:
$$
\mathcal{B}_{i,j}^{[\textsf{s}_{i,j}]} = \int_{e_i}^{e_{j}} \rmd r^{[\textsf{s}_{i,j}]},
$$
where $\rmd r^{[\textsf{s}_{i,j}]}$, defined by \eqref{coHDifCg}, is taken with the sign $\textsf{s}_{i,j}$.
Then the second kind period matrices are  
\begin{gather}
\eta_k = 2\mathcal{B}_{2k-1, 2k}^{[\textsf{s}_{2k-1,2k}]} ,\qquad\qquad  
\eta'_k = - 2\sum_{j=k}^g \mathcal{B}_{2j, 2j+1}^{[\textsf{s}_{2j,2j+1}]} .
\end{gather}
The symmetric matrix responsible for modular invariance\footnote{By introducing the exponential factor
with $\varkappa$, the  theta function is transformed into the modular invariant sigma function, cf.\,\eqref{SigmaThetaRel}.} is
$$\varkappa = \eta \omega^{-1}.$$

The four matrices $\omega$, $\omega'$, $\eta$, $\eta'$ satisfy the Legendre relation \eqref{LegRel},
which serve for verification.

\subsection{Example 1: Arbitrary complex branch points}\label{E:C29}
Let a hyperelliptic curve
of genus $4$ possess the given finite branch points:
\begin{gather*}
-18 - 2\imath,\  -16 + 5 \imath,\  -11 + 3 \imath,\ -10 - \imath,\ -4 + 2 \imath,\ -3 + 3 \imath,\ 
3 + 3 \imath,\ 7 -  2 \imath,\ 13 - \imath.
\end{gather*}
On the other hand, this curve is defined by the equation
\begin{subequations}\label{HypCExmpl}
\begin{align}
&0 = f(x,y) \equiv -y^2 + \mathcal{P}(x),\\
&\mathcal{P}(x) = x^9 + (39 - 10 \imath) x^8 + (217 - 288 \imath) x^7
- (7585 - 826 \imath) x^6 \\
&- (79138 - 82462 \imath) x^5 
+ (324058 + 455846 \imath) x^4  + (4126332 - 3930980 \imath) x^3 \notag  \\
&- (14219032 + 29444932 \imath) x^2 
- (131012592 - 28208616 \imath) x \notag \\
&-101860560 +  245519280 \imath.   \notag
\end{align}
\end{subequations}

Cuts and homology cycles, see  fig.\,\ref{f:PathCycles}, are introduced as explained in subsection~\ref{ss:Hom}.
\begin{figure}[h]
\includegraphics[width=0.6\textwidth]{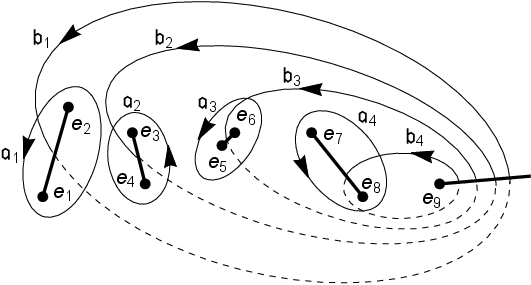}
\caption{Canonical homology cycles.$\phantom{mmmmmmmmmmmm}$}\label{f:PathCycles}
\end{figure}
On fig.\,\ref{f:BPCuts}, the contour  $\Gamma$ is shown in blue. 
Note, that the cut $(e_9,\infty)$ coincides with segment $\Gamma_9$.
The monodromy path $\gamma$ is marked in orange. It goes below points $e_i$.
A curved orange line near a cut shows on which side of the cut $\gamma$ goes.
\begin{figure}[h]
\includegraphics[width=0.43\textwidth]{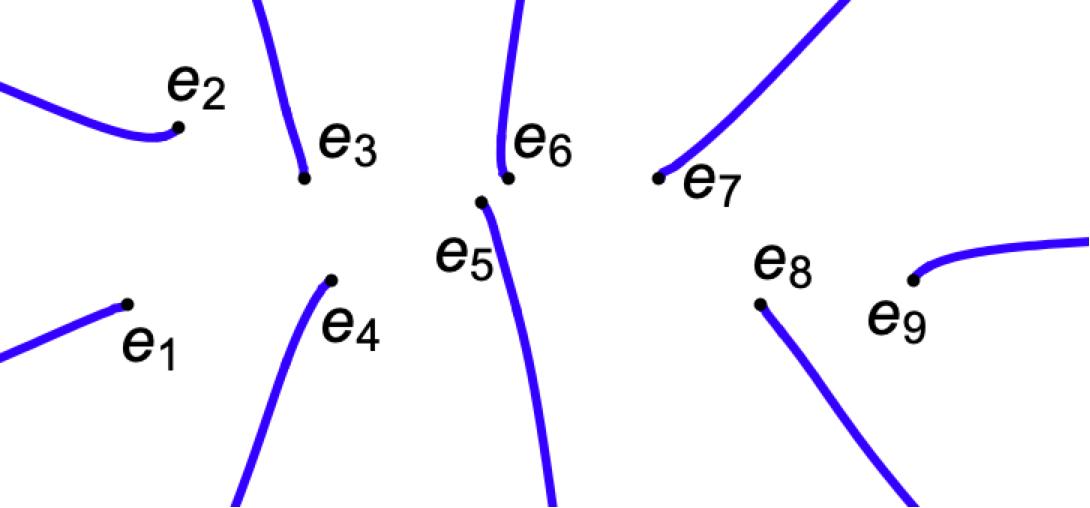}$\quad$
\includegraphics[width=0.43\textwidth]{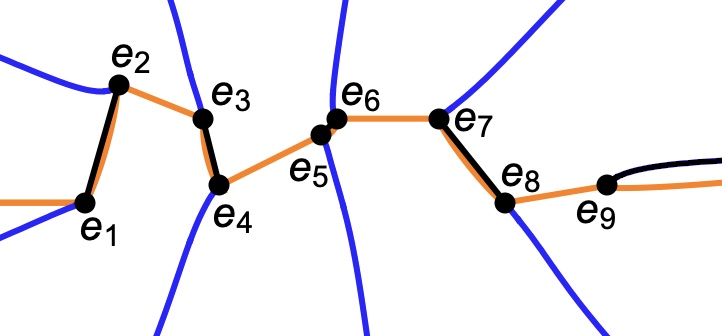} 
\caption{The contour $\Gamma$ (blue), and a continuous path (orange).}\label{f:BPCuts}
\end{figure}

As seen on fig.\,\ref{f:BPArg}, along the contour $\Gamma$ 
solution $y_+$ continuously connects to $y_-$,
and vice versa.
\begin{figure}[h]
\includegraphics[width=0.42\textwidth]{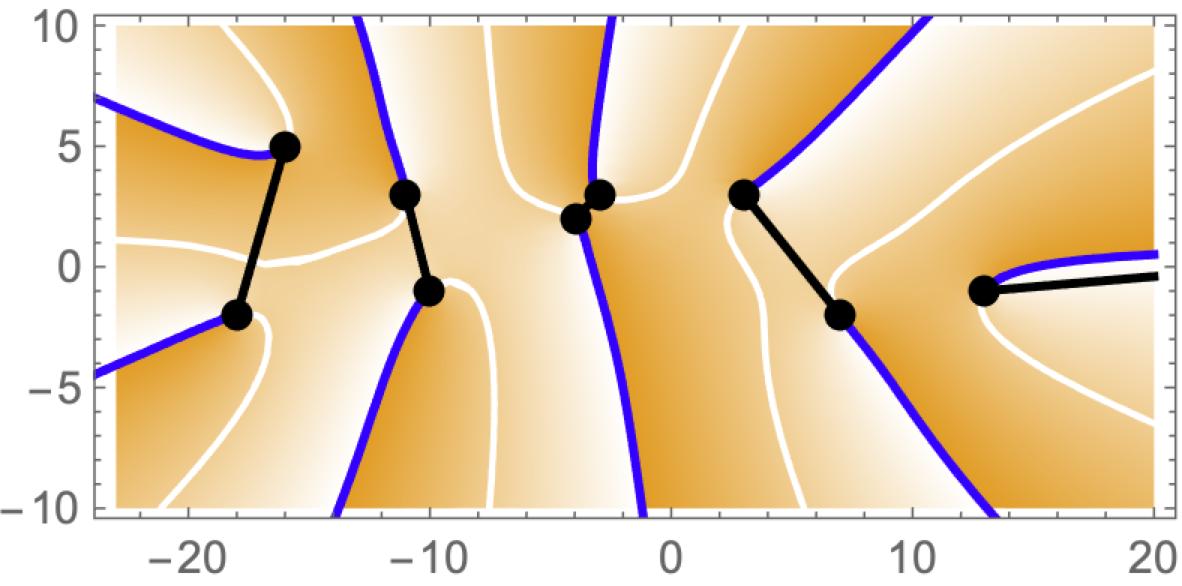}
\includegraphics[width=0.06\textwidth]{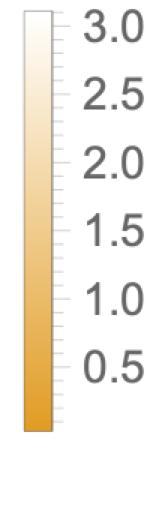}\ 
\includegraphics[width=0.42\textwidth]{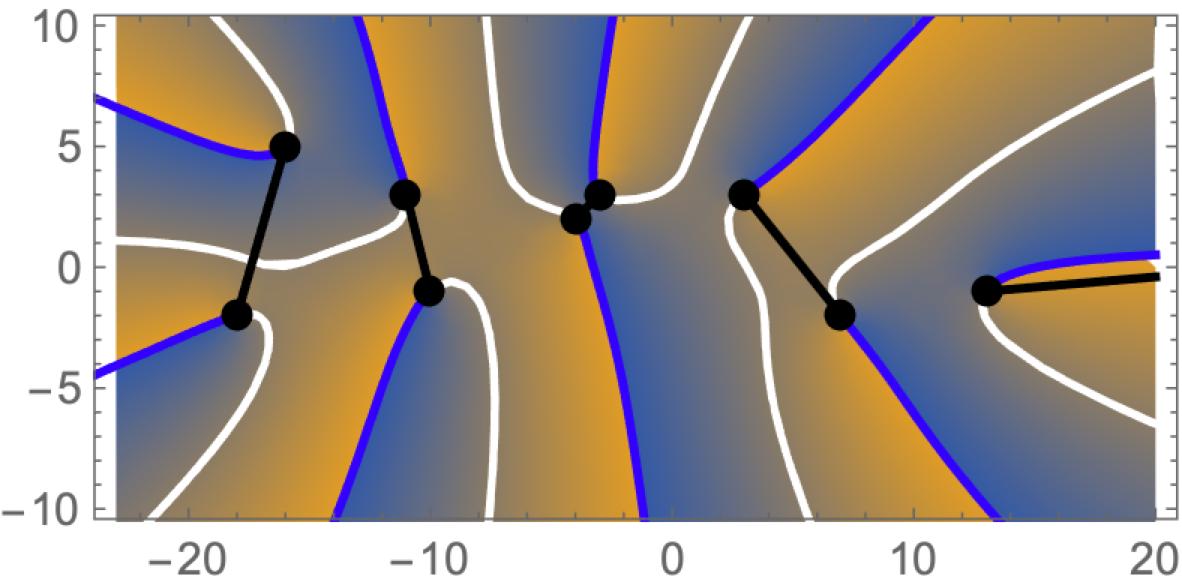}
\includegraphics[width=0.07\textwidth]{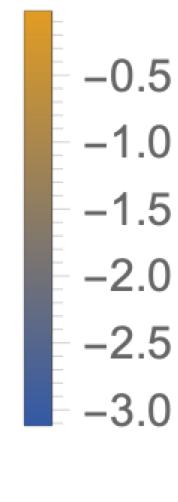}\\
 (a) $\arg y_+$ $\phantom{mmmmmmmmmmmmmm}$ (b) $\arg y_-$ $\phantom{m}$
\caption{Density plots of $\arg y_{\textsf{s}}$, and the contour $\Gamma$ (blue).}\label{f:BPArg}
\end{figure}

Now we mark Sheet\;\textsf{a}, see fig.\,\ref{f:Sheet}.
\begin{figure}[h]
\parbox[b]{0.45\textwidth}{\includegraphics[width=0.45\textwidth]{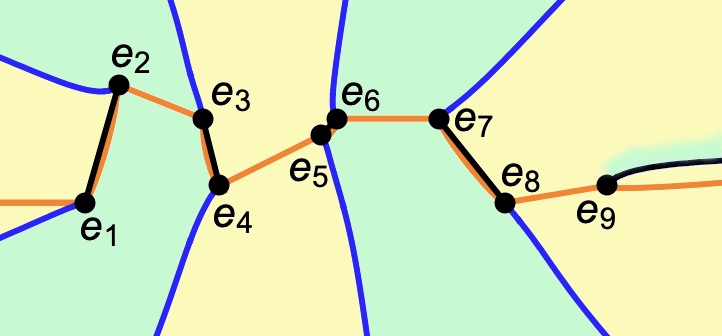}}\ \ 
\parbox[b]{0.07\textwidth}{\includegraphics[width=0.07\textwidth]{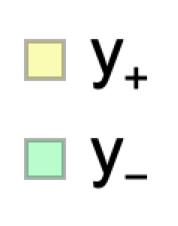}\\
$\quad$ \\ $\quad$ \\ $\quad$}
\caption{Connection of solutions $y_+$ and $y_-$ on Sheet\;\textsf{a}.}\label{f:Sheet}
\end{figure}
Let start with $\textsf{s}_{0,1}\,{=}\,{+}1$ on segment $(-\infty,\, e_1]$. 
The path $\gamma$ (orange) intersects the blue contour
at point $e_1$, and so the sign changes:   $\textsf{s}_{1,2} = -1$. 
The path goes along the right side of the cut $(e_1,e_2)$, 
and the sign remains the same on the next segment: $\textsf{s}_{2,3} =-1$.
From $e_3$ to $e_4$ the path goes along the left side of the cut $(e_3,e_4)$, 
so the sign remains the same: $\textsf{s}_{3,4} = -1$.
The next intersection with the blue contour occurs at point $e_4$, thus $\textsf{s}_{4,5} = +1$.
And again at point $e_5$, so $\textsf{s}_{5,6} = \textsf{s}_{6,7} = \textsf{s}_{7,8} = -1$.
Once again, the path $\gamma$ intersects the blue contour at point~$e_8$, thus $\textsf{s}_{8,9} = \textsf{s}_{9,0} = +1$.
Finally, the sequence of signs is
 $$
 \begin{array}{lcccccccccc}&\{\textsf{s}_{0,1},& \textsf{s}_{1,2},& \textsf{s}_{2,3},& \textsf{s}_{3,4},& \textsf{s}_{4,5},&
 \textsf{s}_{5,6},& \textsf{s}_{6,7},& \textsf{s}_{7,8},& \textsf{s}_{8,9},& \textsf{s}_{9,0}\} =\\
 \text{Sheet\;\textsf{a}:}& \{+1, & -1, & -1, & -1, & +1, & -1,& -1,& -1, &+1, &+1\}.
  \end{array} $$
Along segment $\Gamma_9$ solution $y_+$ on Sheet\;\textsf{a}
connects to $y_-$ on Sheet\;\textsf{b},
the latter is shown in green above of $\Gamma_9$, see fig.\,\ref{f:Sheet}.

First kind differentials as functions of $x$ are 
\begin{gather*}
\rmd u^{[\textsf{s}]} 
= \begin{pmatrix}  x^3 \\ x^2 \\ x \\ 1 \end{pmatrix}
\frac{\rmd x}{- 2 \, \textsf{s}\, \sqrt{\mathcal{P}(x)}}.
\end{gather*} 

According to the picture of homology cycles, see fig.\,\ref{f:PathCycles},
the first kind period matrices $\omega$ and $\omega'$ are computed as follows, cf. \eqref{FKPerN},
\begin{align*}
&\omega = 2 \big(\mathcal{A}_{1,2}^{[-]},\, \mathcal{A}_{3,4}^{[-]},\, \mathcal{A}_{5,6}^{[-]},\, \mathcal{A}_{7,8}^{[-]}\big),\\
&\omega' = -2 \big(\mathcal{A}_{2,3}^{[-]}+\mathcal{A}_{4,5}^{[+]}+\mathcal{A}_{6,7}^{[-]}+\mathcal{A}_{8,9}^{[+]},\,
\mathcal{A}_{4,5}^{[+]}+\mathcal{A}_{6,7}^{[-]}+\mathcal{A}_{8,9}^{[+]},\,
\mathcal{A}_{6,7}^{[-]}+\mathcal{A}_{8,9}^{[+]},\, \mathcal{A}_{8,9}^{[+]}\big).
\end{align*}
Numerical computation of $\mathcal{A}_{i,j}^{[\textsf{s}_{i,j}]}$ is performed with the help of
function \texttt{NIntegrate} with a \texttt{WorkingPrecision} of $18$.
The relations \eqref{HyperInvRels} are satisfied with an accuracy of $10^{-16}$.
On the curve \eqref{HypCExmpl} we obtain
\begin{align*}
&\omega \approx \left( \begin{matrix}
-1.303573 + 0.207439 \imath &
  0.848115 - 0.306788 \imath \\
 0.073367 - 0.003075 \imath &
   -0.083799 + 0.019801 \imath \\
 -0.003985  - 0.000333  \imath &
  0.008037 - 0.001042  \imath \\
 0.000208 + 0.000046 \imath & 
 -0.000751 + 0.000028  \imath \\
\end{matrix} \right.  \phantom{mmmmmmmm} \\
&\phantom{mmmmmmmmmmmmm} \left.
\begin{matrix}
  0.0166625 + 0.063503 \imath & 
   -0.035439 - 0.017840 \imath \\
  0.005372  - 0.014707 \imath &
   -0.007363 - 0.005200 \imath \\
  -0.003023 + 0.002108  \imath &
   -0.001651 - 0.000958 \imath \\
  0.000856 + 0.000006  \imath & 
  -0.000350  - 0.000098 \imath
\end{matrix} \right),\\
&\omega' \approx \left( \begin{matrix}
 -0.604960 - 0.930374 \imath &
  -0.085127 + 0.237963 \imath \\
  0.029948  + 0.024444  \imath &
  0.019758 - 0.067832 \imath \\
 -0.001525 - 0.000835 \imath &
 -0.002854 + 0.005901 \imath \\
  0.000078 + 0.000029 \imath &
  0.000324- 0.000426 \imath 
\end{matrix} \right.  \phantom{mmmmmmmm} \\
&\phantom{mmmmmmmmmmmmm} \left. \begin{matrix}
  0.042341 - 0.110642  \imath &
  0.047121 - 0.129831 \imath \\
  0.011312 - 0.018512 \imath &
  0.006638 - 0.011849 \imath \\
  -0.002257 - 0.001647 \imath &
  0.000894 - 0.001058  \imath \\
  0.000178 + 0.000796  \imath &
  0.000117 - 0.000089 \imath
\end{matrix} \right).
\end{align*}
The corresponding normalized period matrix from the Siegel upper half-space is
\begin{gather*}
\tau \approx \left( \begin{matrix}
 0.416960 + 1.348235 \imath &
 -0.019631 + 0.866637 \imath \\
 -0.019631 + 0.866637 \imath &
 -0.401986 + 1.468494 \imath \\
 0.043442 + 0.592788 \imath &
  0.090347 + 0.771653 \imath \\
 0.013536 + 0.360353 \imath &
  0.020075 + 0.430424 \imath 
\end{matrix} \right. \phantom{mmmmmmmmmmm} \\
\left.  \phantom{mmmmmmmmmmmmm}
\begin{matrix}
  0.043442 + 0.592788 \imath &
  0.013536 + 0.360353 \imath \\
  0.090347+ 0.771653 \imath &
   0.020075 + 0.430424 \imath \\
  0.276110 + 1.677311 \imath & 
  -0.019449 + 0.549477 \imath \\
  -0.019449 + 0.549477 \imath &
  -0.241045 + 0.959518  \imath
\end{matrix} \right).
\end{gather*}
The symmetric property of $\tau$ is satisfied with an accuracy of $10^{-15}$.

With the second kind differentials
\begin{gather}\label{Dif2Hyp}
 \rmd r_{2i-1}^{[\textsf{s}]}  = \frac{R_{2i-1}(x)\,\rmd x}{-2 \textsf{s}\, \sqrt{\mathcal{P}(x)}},\quad i=1,\,2,\,3,\,4, 
\end{gather}
\begin{align*}
&R_{1} = x^4, \\
&R_{3} = 3 x^5 +  (78 - 20 \imath) x^4 + (217 - 288 \imath) x^3, \\
&R_{5} = 5 x^6 + (156 - 40 \imath) x^5   + (651 -  864 \imath) x^4 
- (15170 - 1652 \imath) x^3 \\
&\qquad\qquad  - (79138 - 82462 \imath) x^2,\\
&R_{7} = 7 x^7 + (234 - 60 \imath) x^6 + (1085 - 1440 \imath) x^5  - (30340 - 3304 \imath) x^4 \\
&\qquad\qquad  - (237414 -  247386 \imath) x^3   + (648116 + 911692 \imath) x^2 \\
&\qquad\qquad + (4126332 - 3930980 \imath) x,
 \end{align*}
second kind periods are computed:
\begin{align*}
&\eta \approx \left( \begin{matrix}
22.428062 - 6.098673 \imath & -8.281570 + 4.260894 \imath \\
280.811215 - 73.921437 \imath & -233.173130 + 22.293105 \imath \\
910.855655 + 10.721526 \imath & -2603.233939 - 813.570166 \imath \\
1224.707652 + 409.922637 \imath & -7484.857429 - 1328.792554 \imath 
\end{matrix} \right. \phantom{mmmmm}\\ 
&\phantom{mmmmmmmmmm}
\left. \begin{matrix}
-0.205360 - 0.177582 \imath & -0.198185 -  0.006055 \imath \\
 5.630286 - 0.370818 \imath & -32.028810 + 11.479858 \imath \\
 299.647570 + 719.376216 \imath & 1242.433044 + 350.238889 \imath \\
  487.787097 + 8861.332415 \imath & 6189.443173 - 7889.884228 \imath
\end{matrix} \right),\\
&\eta' \approx  \left(\begin{matrix}
13.855811 + 9.185567 \imath & 1.880170 - 4.181065 \imath \\
150.913701 + 40.841618 \imath & 64.433701 - 272.113911 \imath \\
420.612695 + 72.624802 \imath &1188.248232 - 1359.820959 \imath \\
504.296929 + 98.209496 \imath & 3559.979547 - 2711.948811 \imath 
\end{matrix} \right. \phantom{mmmmm}\\ 
&\phantom{mmmmmmmm}
\left. \begin{matrix}
 0.264151 - 1.540944 \imath & 0.305811 - 1.421979 \imath \\
 -26.525723 - 208.762530 \imath & -27.152420 - 206.017 \imath \\
 -1290.627353 - 496.867943 \imath&-1470.266735 - 1061.435845 \imath \\
 -9584.859568 + 9412.897457 \imath & -2002.776466 + 2196.979511 \imath
\end{matrix} \right),
\end{align*}
and the matrix
\begin{gather*}
\varkappa \approx \left( \begin{matrix}
-26.150273 + 5.226639 \imath & -113.639362 + 91.745099 \imath \\
-113.639362 + 91.745099 \imath & 2527.918193 + 333.2000001 \imath \\
 815.048336 + 59.142845 \imath & 6691.213749 - 15142.962600 \imath \\
 2796.548807 - 2715.208601 \imath & -19805.451622 - 22245.716646 \imath 
\end{matrix} \right. \phantom{mmmmm}\\ 
\phantom{mm}
\left. \begin{matrix}
 815.048336 + 59.142845 \imath & 2796.548807 - 2715.208601 \imath \\
 6691.213749 - 15142.962600 \imath & -19805.451622 - 22245.716646 \imath \\
-501204.576087 - 151451.871496 \imath & -1572965.591976 + 1699015.043174 \imath \\
-1572965.591976 + 1699015.043174 \imath & -403196.119224 + 19865411.502694 \imath
\end{matrix} \right).
\end{gather*}
The symmetric property of $\varkappa$ is accurate up to  $10^{-8}$, and
the Legendre relation  up to $10^{-14}$.

\bigskip
\textbf{Verification.}
An analog of the Bolza formulas on a genus $4$ hyperelliptic curve, see
\cite[Eq.\,(40)]{BerTF2020}, is given by
\begin{equation}\label{BolzaFHE} 
e_\iota = - \frac{\partial_{u_1,u_7}^2 \theta[\{\iota\}](\omega^{-1} u)}
  {\partial_{u_1,u_5}^2 \theta[\{\iota\}](\omega^{-1} u)} \Big|_{u=0}. 
\end{equation}
  According to the chosen homology basis, we have the following correspondence between 
  characteristics and branch points:
\begin{align*}
&e_1 = -18-2\imath& &\small [\varepsilon_1] = \begin{pmatrix} 1&0&0&0 \\ 0&0 &0 &0 \end{pmatrix} &
& [\{1\}]\small = \begin{pmatrix} 0&1&1&1 \\ 0&1 &0 &1 \end{pmatrix},&\\
&e_2 = -16+5\imath& &\small [\varepsilon_2] = \begin{pmatrix} 1&0&0&0 \\ 1&0 &0 &0 \end{pmatrix} &
& [\{2\}]\small = \begin{pmatrix} 0&1&1&1 \\ 1&1 &0 &1 \end{pmatrix}, &\\
&e_3 = -11+3\imath& &\small [\varepsilon_3] =\begin{pmatrix} 0&1&0&0 \\ 1&0 &0 &0 \end{pmatrix} &
& [\{3\}]\small =  \begin{pmatrix} 1&0&1&1 \\ 1&1 &0 &1 \end{pmatrix}, &\\
&e_4 = -10- \imath& &\small [\varepsilon_4] = \begin{pmatrix} 0&1&0&0 \\ 1&1 &0 &0 \end{pmatrix} &
& [\{4\}]\small = \begin{pmatrix} 1&0&1&1 \\ 1&0 &0 &1 \end{pmatrix}, &\\
&e_5 = -4 +2\imath& &\small [\varepsilon_5] = \begin{pmatrix} 0&0&1&0 \\ 1&1 &0 &0 \end{pmatrix} &
& [\{5\}]\small  = \begin{pmatrix} 1&1&0&1 \\ 1&0 &0 &1 \end{pmatrix}, &\\
&e_6 = -3+3\imath& &\small [\varepsilon_6] = \begin{pmatrix} 0&0&1&0 \\ 1&1 &1 &0 \end{pmatrix} &
& [\{6\}]\small = \begin{pmatrix} 1&1&0&1 \\ 1&0 &1 &1 \end{pmatrix},  &\\
&e_7 =  3+3\imath& &\small [\varepsilon_7] = \begin{pmatrix} 0&0&0&1 \\ 1&1 &1 &0 \end{pmatrix} &
& [\{7\}]\small = \begin{pmatrix} 1&1&1&0 \\ 1&0 &1 &1 \end{pmatrix}, &\\
&e_8 =  7-2\imath& &\small [\varepsilon_8] = \begin{pmatrix} 0&0&0&1 \\ 1&1 &1 &1 \end{pmatrix} &
& [\{8\}]\small = \begin{pmatrix} 1&1&1&0 \\ 1&0 &1 &0 \end{pmatrix}, &\\
&e_9 =  13- \imath& &\small [\varepsilon_9] = \begin{pmatrix} 0&0&0&0 \\ 1&1 &1 &1 \end{pmatrix} &
& [\{9\}]\small = \begin{pmatrix} 1&1&1&1 \\ 1&0 &1 &0 \end{pmatrix},& 
\end{align*}
where
$[\{\iota\}] = [\varepsilon_\iota] + [K] $, and
\begin{equation}\label{KChHyp}
[K] = \sum_{i=1}^4 [\varepsilon_{2i}] = \small \begin{pmatrix} 1&1&1&1 \\ 0&1 &0 &1 \end{pmatrix}. 
 \end{equation}
The zero matrix serves as the characteristic $[\varepsilon_0]$ of
the branch point $e_0$ at infinity.
The method of computing characteristics is adopted from \cite[p.\,1012]{ER2008}.

The formulas \eqref{BolzaFHE} are satisfied with an accuracy of $10^{-14}$.

\begin{rem}
Period matrices of the first kind obtained for the curve \eqref{HypCExmpl}
with the help of \texttt{algcurves} differ from the presented results, 
as well as a different and much more tangled homology basis is chosen for computation.
The period matrices obtained from \texttt{algcurves} satisfy the Bolza formulas for nine half-integer
characteristics. A correspondence between characteristics and branch points is not clear from the homology basis
chosen for computation,  though one can discover this correspondence
by applying the Bolza formulas  to all characteristics.
\end{rem}

\subsection{Example 2: Real branch points}\label{E:C29R}
Consider briefly the case of a curve \eqref{HypCe} with all real branch points:
\begin{gather*}
-18, -15, -11, -5, 1, 2, 7, 12, 16.
\end{gather*}
Such a curve is defined by the equation
\begin{multline}\label{HypCE2}
0 = -y^2 + x^9 + 11 x^8 - 514 x^7 - 4602 x^6 + 82441 x^5 + 506395 x^4 \\
- 4495768 x^3  - 11079084 x^2 + 54907920 x - 39916800. 
\end{multline}

\begin{figure}[h]
\includegraphics[width=0.46\textwidth]{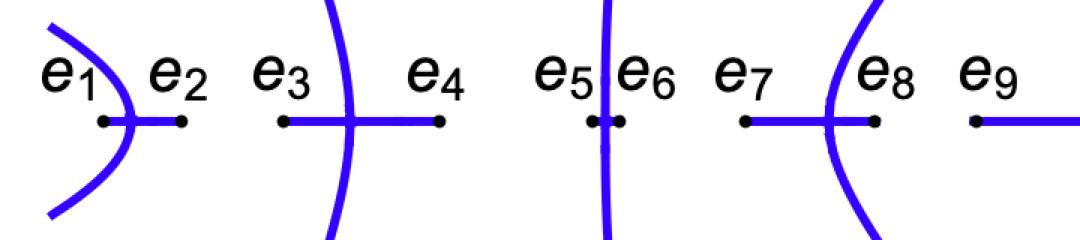}$\ $\\ $\quad$ \\
\includegraphics[width=0.44\textwidth]{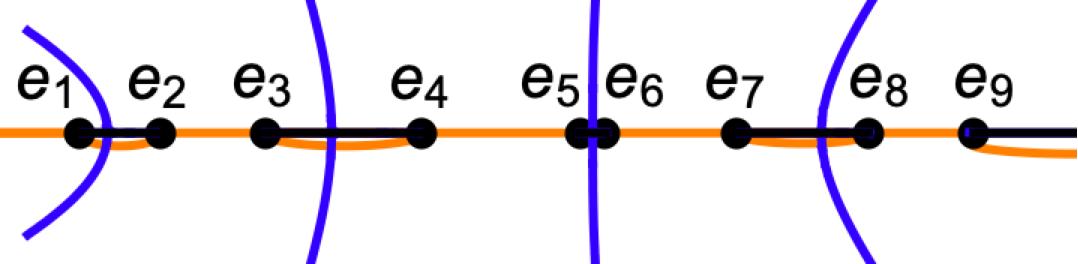} 
\caption{The contour $\Gamma$ (blue), and a continuous path $\gamma$ (orange).}\label{f:BPCutsReal}
\end{figure}
The contour $\Gamma$  and the monodromy path $\gamma$ through all branch points 
are shown on fig.\,\ref{f:BPCutsReal}.
In this case,  the sequence of signs has a clear pattern:
$$
 \begin{array}{cccccccccc}\{\textsf{s}_{0,1},& \textsf{s}_{1,2},& \textsf{s}_{2,3},& \textsf{s}_{3,4},& \textsf{s}_{4,5},&
 \textsf{s}_{5,6},& \textsf{s}_{6,7},& \textsf{s}_{7,8},& \textsf{s}_{8,9},& \textsf{s}_{9,0}\} =\\
  \{+1, & -1, & -1, & +1, & +1, & -1,& -1,& +1, &+1, &-1\}.\ \ 
  \end{array} $$
\begin{rem}
Note, that  cuts coincide with horizontal segments of the contour $\Gamma$,
and so the sign changes when the left end-point of a cut is reached, that is 
 at every $e_{2k-1}$, $k=1$, \ldots, $g+1$. This implies, $\textsf{s}_{4k-3,4k-2} = \textsf{s}_{4k-2,4k-1}  = -1$,
  $k=1$, \ldots, $[(g+1)/2]$,
 and $\textsf{s}_{4k-1,4k} = \textsf{s}_{4k,4k+1}  = +1$, $k=1$, \ldots, $[g/2]$, since
 we always start with $\textsf{s}_{0,1} = +1$. The final segment has the sign $\textsf{s}_{2g+1,0} = -1$
 if the genus is even, or $\textsf{s}_{2g+1,0} = +1$  if the genus is odd.
\end{rem}

First kind not normalized periods are 
\begin{align*}
&\omega \approx \begin{pmatrix}
-0.675637& 0.287434 & 0.002651& -0.309937 \\
0.041299 & -0.032471 & 0.001599 & -0.031911 \\
-0.002535 & 0.003908 & 0.001010 & -0.003404 \\
0.000156 & -0.000507 & 0.000673 & -0.000377
\end{pmatrix},\\
&\omega' \approx \begin{pmatrix}
-0.939638 \imath & -0.289356 \imath & -0.312311 \imath & -0.394964 \imath \\
0.030634 \imath & -0.018895 \imath & -0.013110 \imath & -0.028577 \imath \\
-0.001371 \imath & 0.002447 \imath & 0.001224 \imath & -0.002089 \imath \\
0.000066 \imath & -0.000232 \imath & 0.000698 \imath & -0.000154 \imath
\end{pmatrix}.
\end{align*}
Then the normalized period matrix from the Siegel upper half-space is
\begin{align*}
&\tau \approx \begin{pmatrix}
1.602330 \imath & 0.820786 \imath & 0.514534 \imath & 0.304355 \imath \\
0.820786 \imath & 1.304404 \imath & 0.648249 \imath & 0.359593 \imath \\
0.514534 \imath & 0.648249 \imath & 1.686169 \imath & 0.501628 \imath \\
0.304355 \imath & 0.359593 \imath & 0.501628 \imath & 0.948644 \imath
\end{pmatrix}.
\end{align*}

\begin{rem}
When all branch points are real, the monodromy path $\gamma$ coincides with the real axis. 
Moreover,  $\mathcal{P}$ is real-valued along $\gamma$, positive on segments $[e_{2k-1},e_{2k}]$,
$k=1$, \ldots, $g+1$, where cuts are made, and negative 
on the remaining segments $[e_{2k},e_{2k+1}]$, $k=0$, \ldots, $g$. If one branch point of a curve is located at infinity,
then $e_0$ stands for $-\infty$, end $e_{2g+2}$ stands for $\infty$. If all branch points are finite, then
$e_{2g+2}$ denotes the same point as $e_0$, which is the smallest one.

Since $\mathcal{P}$ has an alternating sign along the path, all $\mathfrak{a}$-periods are real,
and all $\mathfrak{b}$-periods are purely imaginary. Moreover, 
\begin{equation}\label{VZ}
y(x) = (-\imath)^{j} \sqrt{|\mathcal{P}(x)|}  \qquad
\text{on\ }\ [e_j,e_{j+1}],\quad  j=0,\, \ldots,\, 2g+1.
\end{equation}
The formula \eqref{VZ} was discovered by V. Enolsky.
It shows how $y_+$ and $y_-$ connect on Sheet\;\textsf{a}.
\end{rem}

Second kind periods on the curve \eqref{HypCE2} computed by \eqref{Dif2Hyp} with
\begin{align*}
&R_{1} = x^4, \\
&R_{3} = 3 x^5 +  22 x^4  - 514x^3, \\
&R_{5} = 5 x^6 + 44 x^5   - 1542 x^4 - 9204 x^3  + 82441 x^2,\\
&R_{7} = 7 x^7 + 66 x^6 - 2570 x^5  - 18408 x^4 \\
&\qquad\qquad  + 247323 x^3   + 1012790 x^2  -4495768 x
 \end{align*}
are
\begin{align*}
&\eta \approx \begin{pmatrix}
11.099150 & -2.673068 &  0.004562 & -3.109806 \\
42.176466 & -129.207930 & -1.238236 & -5.426458 \\
-383.268207 & -1342.343715 & 100.817590 & 1906.755187 \\
-578.759149 &  4760.249993 & -2368.318297 & -1963.133866 
\end{pmatrix},\\
&\eta' \approx \begin{pmatrix}
3.501174 \imath & -5.136174 \imath & -5.034880 \imath & -5.515607 \imath \\
-11.2650450 \imath & -187.532997 \imath & -174.869192 \imath & -151.780483 \imath \\
-195.044640 \imath & 288.885494 \imath &  811.312906  \imath & 816.306629 \imath \\ 
-195.450172 \imath & 1246.674997 \imath & 6300.606235 \imath & -562.511001 \imath
\end{pmatrix},
\end{align*}
and the symmetric matrix is
\begin{gather*}
\varkappa \approx \begin{pmatrix}
-13.123159 & 129.285113 &  1107.820797 & -1910.386399 \\
 129.285113 &  1362.173530 & -26772.601447 & 34575.690532 \\
  1107.820797 & -26772.601447 & -519356.757226 & 988034.553637 \\
  -1910.386399 &  34575.690532 & 988034.553637 & -5074619.889795 
\end{pmatrix}.
\end{gather*}

Numerical integration is performed with a \texttt{WorkingPrecision} of $18$. 
The same accuracy as in Example~1 is achieved.

\section{Computation of $\wp$-functions on a hyperelliptic curve}\label{s:HypWP}
We start with computing the Abel image $\mathcal{A}(P)$ of a given point $P$ by \eqref{AbelM},
where the standard not normalized holomorphic differentials \eqref{HDifCg} are used.
A choice of the path from the base-point to $P$ is tightly attached to 
the Riemann surface constructed in the previous section. 
Below, an explanation how to choose
such a path correctly is given.
The Abel image of a  divisor $D =  \sum_{i=1}^n P_i$ is computed by
\begin{equation}\label{AMapD}
 \mathcal{A} (D) = \textstyle \sum_{i=1}^n \mathcal{A}(P_i).
 \end{equation}
 The Abel image of any divisor can be computed in this way.

$\wp$-Functions are defined on non-special divisors only. 
We compose non-special divisors $D$ as positive divisors of degree $n\geqslant g$ with no pairs of points in involution
on a hyperelliptic curve of genus $g$. $\wp$-Functions are calculated at $u=\mathcal{A} (D)$ by means of \eqref{WPdef},
where the characteristic $[K]$ of the vector of Riemann constants is required.
 
On a hyperelliptic curve, $[K]$
is computed as a sum of odd characteristics corresponding to branch points, cf.\,\eqref{KChHyp}.
The same $[K]$ is obtained from the formula \eqref{VRC}.
The function $\theta[K]$  has the maximal  
order of vanishing at $u=0$ computed with respect to 
 the Sato weight.  The  order of vanishing\footnote{See computation of the order of vanishing of $\theta[K]$
 with respect to not normalized coordinates in 
 \texttt{https://community.wolfram.com/groups/-/m/t/3296279}, where $[K]$ 
 are computed for Examples 1 and 3 from this paper.} 
 coincides with 
 the negative Sato weight of $\sigma$-function, which is $-\wgt \sigma = \frac{1}{2} (g+1) g$
 on a hyperelliptic curve of genus $g$, for more details see Theorem~\ref{T:Kval} in section~\ref{s:TrigWP}.

Let $P_i =(x_i,y_i)$. 
We choose $e_i$ close to $x_i$, such that  $y_{\textsf{s}}$
does not change the sign over the segment $[e_i,x_i]$.
The sheet where $P_i$ is located is identified
from comparing the sign $\textsf{n}$ of $y_i = y_\textsf{n}(x_i)$
with the sign  on  $[e_i,x_i]$ according to the sequence of signs on Sheet\;\textsf{a}.
A path to $P_i$ is drawn on the sheet where the point is located.
The path starts at $-\infty$ on the real axis, and goes to $e_i$ 
along $\gamma^{\textsf{n}}$, the lift of $\gamma$ on Sheet~$\textsf{n}$.
Then the segment $[e_i,x_i]$ on the same sheet is added to this path. 
Along the path to $P_i$ the Abel image $\mathcal{A}(P_i)$ is computed.

The proposed algorithm is illustrated by examples.
Divisors of degree equal to the genus $g$ of a curve are considered, and so 
the Jacobi inversion problem is used for verification.
One can compute $\wp$-functions on a divisor of degree
greater than~$g$, and  solve the Jacobi inversion problem for the
corresponding reduced divisor.
The latter serve as an implementation of addition on a curve.

\subsection{Example 1a}\label{P:C29}
We continue to work with the curve \eqref{HypCExmpl} from Example 1 in the previous section.

Let $u$ be the Abel image of a non-special divisor $D = \sum_{i=1}^4 P_i$,  
\begin{equation}\label{Ex1aPs}
\begin{split}
P_1 = (x_1,y_1) &= (-9+\imath, -8\sqrt{-918\,645-541\,515 \imath}) \\
&\approx (-9+\imath, -2174.219935 + 7969.975679 \imath),\\
P_2 = (x_2,y_2) &= (-4-3\imath, -20\sqrt{924613 - 1261876 \imath})  \\
&\approx (-4-3\imath, -22311.336861 + 11311.522997 \imath),\\
P_3 = (x_3,y_3) &= (1+2\imath, 10\sqrt{-1744002 + 734019 \imath}) \\
&\approx (1+2\imath, 2721.885087 + 13483.651524 \imath),\\
P_4 = (x_4,y_4) &= (6+4\imath, -4\sqrt{99702405 - 110095815 \imath})\\
& \approx (6+4\imath, -44563.130818 + 19764.466810 \imath).
\end{split}
\end{equation}
\begin{figure}[h]
\parbox[b]{0.43\textwidth}{\includegraphics[width=0.43\textwidth]{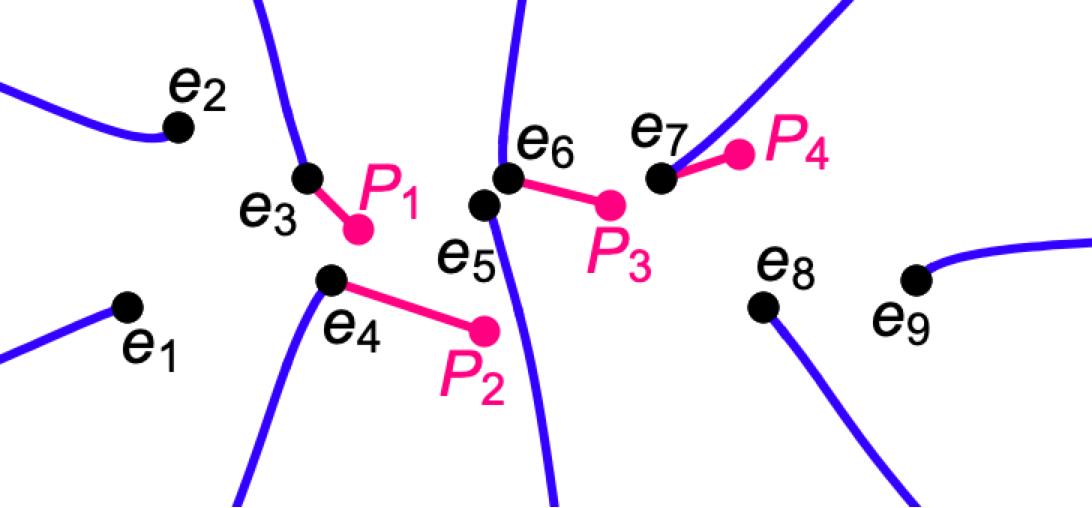}}\ \ \
\parbox[b]{0.43\textwidth}{\includegraphics[width=0.43\textwidth]{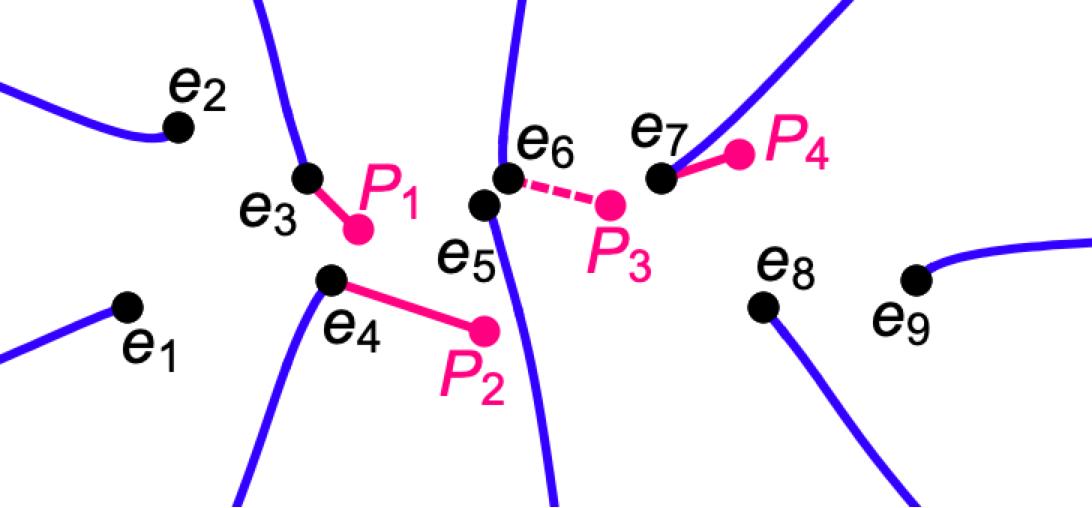}}\ 
\parbox[b]{0.07\textwidth}{\includegraphics[width=0.07\textwidth]{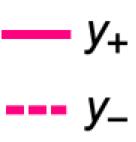}\\
$\quad$ \\ $\quad$ \\ $\quad$ \\ $\quad$}\\
\parbox[b]{0.43\textwidth}{\centering (a) Example 1a}\ \ \ 
\parbox[b]{0.43\textwidth}{\centering (b) Example 1b}$\quad\quad\quad$
\caption{Paths to points of divisor $D$.}\label{f:HEEx1}
\end{figure}
Actually, $y_i = y_+(x_i)$ for all $i$. Using the sequence of signs on Sheet\;\textsf{a},  and fig.\,\ref{f:Sheet}, 
we find that $P_1$, $P_2$, $P_4$ are located on Sheet\;\textsf{a},
and $P_3$  on Sheet\;\textsf{b}. Then we compute the corresponding Abel images, 
see locations of points $P_i$ on fig.\,\ref{f:HEEx1}(a)
\begin{align*}
&\mathcal{A}(P_1) = \mathcal{A}_{0,1}^{[+]} + \sum_{i=1}^2 \mathcal{A}_{i,i+1}^{[-]} + \int_{e_3}^{x_1} \rmd u^{[+]},\\
&\mathcal{A}(P_2) = \mathcal{A}_{0,1}^{[+]}  +  \sum_{i=1}^3 \mathcal{A}_{i,i+1}^{[-]} + \int_{e_4}^{x_2} \rmd u^{[+]},\\
&\mathcal{A}(P_3) = - \Big(\mathcal{A}_{0,1}^{[+]} +  \sum_{i=1}^3 \mathcal{A}_{i,i+1}^{[-]} + \mathcal{A}_{4,5}^{[+]} 
+ \mathcal{A}_{5,6}^{[-]} + \int_{e_6}^{x_3} \rmd u^{[-]} \Big),\\
&\mathcal{A}(P_4) = \mathcal{A}_{0,1}^{[+]} 
+  \sum_{i=1}^3 \mathcal{A}_{i,i+1}^{[-]} + \mathcal{A}_{4,5}^{[+]} 
+ \sum_{i=5}^6 \mathcal{A}_{i,i+1}^{[-]} + \int_{e_7}^{x_4} \rmd u^{[+]},
\end{align*}
 and find
$$ u(D) = \sum_{i=1}^4 \mathcal{A}(P_i) \approx \begin{pmatrix}
-1.182750 + 0.205635 \imath \\  0.073714 - 0.038375 \imath\\ 
-0.004762+  0.002674 \imath \\ 0.000592 - 
 0.000064 \imath
\end{pmatrix}. $$
By means of \eqref{WPdef}, with $\theta$-function approximated by
a partial sum of \eqref{ThetaDef}, $n_i \leqslant 5$,
we obtain 
\begin{gather}\label{WPvalsHyp}
\begin{aligned}
&\wp_{1,1}\big(u(D)\big) \approx -6 + 4\imath,& &\wp_{1,1,1}\big(u(D)\big) \approx 91.581255 - 159.929002 \imath, \\
&\wp_{1,3}\big(u(D)\big) \approx 42 + 53\imath,& &\wp_{1,1,3}\big(u(D)\big) \approx 23.849556 - 1665.831810 \imath, \\
&\wp_{1,5}\big(u(D)\big) \approx 193 + 191\imath,& &\wp_{1,1,5}\big(u(D)\big) \approx -6971.970187 - 998.734510 \imath, \\
&\wp_{1,7}\big(u(D)\big) \approx 446 - 578\imath,& &\wp_{1,1,7}\big(u(D)\big) \approx -5733.795768 - 18693.578334  \imath.
\end{aligned}
\end{gather}

Next, we use the solution \eqref{EnC22g1} of the Jacobi inversion problem to verify if
the obtained Abel image corresponds to the given divisor $D$,
and at the same time, to verify
the compliance of the obtained periods with the curve. 
On a hyperelliptic curve of genus 4 this solution acquires the form
\begin{align*}
&\mathcal{R}_{8}(x;u) \equiv x^4 - x^3 \wp_{1,1}(u) - x^2 \wp_{1,3}(u) - x \wp_{1,5}(u) - \wp_{1,7}(u) = 0,\\ 
&\mathcal{R}_{9}(x,y;u) \equiv 2 y + x^3 \wp_{1,1,1}(u) + x^2 \wp_{1,1,3}(u) + x \wp_{1,1,5}(u) + \wp_{1,1,7}(u)= 0.
\end{align*}
The two polynomial functions  $\mathcal{R}_{8}$ and $\mathcal{R}_{9}$ with
coefficients \eqref{WPvalsHyp} both vanish on the given $D$. 
 
On the other hand, the common divisor of zeros of $\mathcal{R}_{8}$ and $\mathcal{R}_{9}$,
which is a degree~$g$ non-special divisor $D_g$,
uniquely defines these two polynomial functions.
With a divisor $D_4 = \sum_{i=1}^4 (x_i,y_i)$
on a genus $4$ hyperelliptic curve, we have
\begin{align*}
&\mathcal{R}_{8}(x;D_4) = 
\frac{\small \begin{vmatrix} 
x^4 & x^3 & x^2 & x & 1 \\
x_1^4 & x_1^3 & x_1^2 & x_1 & 1 \\
x_2^4 & x_2^3 & x_2^2 & x_2 & 1 \\
x_3^4 & x_3^3 & x_3^2 & x_3 & 1 \\
x_4^4 & x_4^3 & x_4^2 & x_4 & 1 \end{vmatrix}}
{\small \begin{vmatrix}
 x_1^3 & x_1^2 & x_1 & 1 \\
 x_2^3 & x_2^2 & x_2 & 1 \\
 x_3^3 & x_3^2 & x_3 & 1 \\
 x_4^3 & x_4^2 & x_4 & 1
\end{vmatrix}},&
&\mathcal{R}_{9}(x,y;D_4) = 2
\frac{\small \begin{vmatrix} 
y & x^3 & x^2 & x & 1 \\
y_1 & x_1^3 & x_1^2 & x_1 & 1 \\
y_2 & x_2^3 & x_2^2 & x_2 & 1 \\
y_2 & x_3^3 & x_3^2 & x_3 & 1 \\
y_3 & x_4^3 & x_4^2 & x_4 & 1
\end{vmatrix}}
{\small \begin{vmatrix} 
 x_1^3 & x_1^2 & x_1 & 1 \\
 x_2^3 & x_2^2 & x_2 & 1 \\
 x_3^3 & x_3^2 & x_3 & 1 \\
 x_4^3 & x_4^2 & x_4 & 1
 \end{vmatrix}}.&
\end{align*}
Then, the given divisor  $D$ defined by \eqref{Ex1aPs} is the common divisor of zeros
of the two polynomial functions
\begin{subequations}\label{R8R9}
\begin{align}
&\mathcal{R}_{8}(x;D) \equiv x^4 + (6-4\imath) x^3 - (42 + 53\imath) x^2 
- (193 + 191\imath) x - 446 + 578\imath, \\
&\mathcal{R}_{9}(x,y;D) \equiv 2 y + (91.581255 - 159.929002 \imath) x^3 \\
&\quad + (23.849556 - 1665.831810) x^2 - (6971.970187 + 998.734510 \imath) x \notag \\
&\phantom{mmmmmmmmmmmmmmmmmm} -5733.795768 - 18693.578334 \imath. \notag
\end{align}
\end{subequations}
Coefficients of $\mathcal{R}_{8}$ and $\mathcal{R}_{9}$ in \eqref{R8R9}  give values of 
 $\wp$-functions at $u(D)$, which coincide with \eqref{WPvalsHyp}
 with an accuracy of $10^{-12}$.

\subsection{Example 1b}
Let  $D$ from Example 1a be slightly modified, by moving $P_3$ from Sheet\;\textsf{b} to Sheet\;\textsf{a},
where $y_3 = y_-(x_3)$. That is 
\begin{align*}
P_3 = (x_3,y_3) &= (1+ 2\imath, -10\sqrt{-1744002 + 734019 \imath}) \\
&\approx (1+2\imath, -2721.885087 - 13483.651524 \imath).
\end{align*}
Then $\mathcal{A}(P_3)$ is computed with the opposite sign:
\begin{align*}
&\mathcal{A}(P_3) = \mathcal{A}_{0,1}^{[+]} +  \sum_{i=1}^3 \mathcal{A}_{i,i+1}^{[-]} + \mathcal{A}_{4,5}^{[+]} 
+ \mathcal{A}_{5,6}^{[-]} + \int_{e_6}^{x_3} \rmd u^{[-]} .
\end{align*}
Thus,
$$ u(D) = \sum_{i=1}^4 \mathcal{A}(P_i) \approx \begin{pmatrix}
-1.574832 + 0.029543  \imath \\ 0.073141 - 0.050092 \imath\\ 
-0.003641 + 0.002439 \imath\\  0.000890 + 0.000121 \imath
\end{pmatrix}. $$
Values of $\wp_{1,2i-1}$, $i=1$, $2$, $3$, $4$, remain the same, within the accuracy.
The new values of $\wp_{1,1,2i-1}$, are
\begin{gather}\label{WPvalsHyp2}
\begin{aligned}
 &\wp_{1,1,1}\big(u(D)\big) \approx -51.390396 - 175.145987 \imath,& \\
 &\wp_{1,1,3}\big(u(D)\big) \approx -1007.385975 - 1486.407403 \imath,& \\
 &\wp_{1,1,5}\big(u(D)\big) \approx -3163.745380 + 5334.829741 \imath,& \\
 &\wp_{1,1,7}\big(u(D)\big) \approx 10094.385116 + 25500.899104 \imath.& 
\end{aligned}
\end{gather}

Evidently,  $\mathcal{R}_{8}$ remains unchanged, since values of $x_i$ are kept the same.
The new function $\mathcal{R}_{9}$ acquires the form
\begin{align*}
&\mathcal{R}_{9} (x,y) = 2 y - (51.390396 + 175.145987 \imath) x^3 \\
&\quad - (1007.385975 + 1486.407403 \imath) x^2 - (3163.745380 - 5334.829741 \imath) x \\
&\quad\phantom{mmmmmmmmmmmmmmmmmm} + 10094.385116 + 25500.899104 \imath,
\end{align*}
and its coefficients coincide with the values \eqref{WPvalsHyp2}
 with an accuracy of $10^{-10}$.

\section{Periods on a trigonal curve}\label{s:TrigPer}

\subsection{Trigonal curves}
A generic trigonal curve is defined by the equation
\begin{gather}\label{TrigCPQTEq}
 0 = \tilde{f}(x,y) = - y^3 + y^2 \mathcal{T}(x) + y\mathcal{Q}(x) + \mathcal{P}(x).
\end{gather}
Let maximal degrees of polynomials $\mathcal{P}$, $\mathcal{Q}$, $\mathcal{T}$ be as shown in the table below. 
The genus $g$ of a curve is displayed in the last column.
\begin{align*}
&\text{Case 1:}& &\deg \mathcal{P} = 3\mFr+1,&  &\deg \mathcal{Q} = 2\mFr,&  &\deg \mathcal{T} = \mFr,& &g = 3\mFr;&\\
&\text{Case 2:}& &\deg \mathcal{P} = 3\mFr+2,&  &\deg \mathcal{Q} = 2\mFr+1,&  &\deg \mathcal{T} = \mFr,& &g = 3\mFr+1;&\\
&\text{Case 3:}& &\deg \mathcal{P} = 3\mFr+3,& &\deg \mathcal{Q} = 2\mFr+2&  &\deg \mathcal{T} = \mFr+1,&  &g = 3\mFr+1.&
\end{align*}
Note, that $y^2 \mathcal{T}(x)$ is eliminated by the map 
$y \mapsto \tilde{y} +\tfrac{1}{3} \mathcal{T}(x)$, which leads to
\begin{gather*}
0 = \tilde{f}(x,\tilde{y})= - \tilde{y}^3 +   \widetilde{\mathcal{Q}}(x) \tilde{y} +  \widetilde{\mathcal{P}}(x),\\
\qquad\qquad\quad \begin{split}
&\widetilde{\mathcal{Q}}(x)  = \mathcal{Q}(x) + \tfrac{1}{3} \mathcal{T}(x)^2, \\
&\widetilde{\mathcal{P}}(x)  = \mathcal{P}(x)  + \tfrac{1}{3} \mathcal{Q}(x) \mathcal{T}(x) + 
\tfrac{2}{27}\mathcal{T}(x)^3.
\end{split}
\end{gather*}
The discriminant of \eqref{TrigCPQTEq} is computed as follows
\begin{gather}
\Delta(x) =  \widetilde{\mathcal{P}}(x)^2 - \tfrac{4}{27} \widetilde{\mathcal{Q}}(x)^3,
\end{gather}
where
\begin{equation}\label{BPNums}
\deg \Delta = N =  \left\{ \begin{array}{ll} 
6 \mFr + 2, & \text{Case 1;} \\
6 \mFr + 4,  & \text{Case 2;} \\
6 \mFr + 6,  & \text{Case 3.} 
\end{array} \right.
\end{equation}
The degree of  $\Delta$ shows the number $N$ of finite branch points $\{B_i = (e_i,h_i)\}_{i=1}^N$.
In Case~$3$, a curve has $6 \mFr + 6$ branch points, all finite.
In Cases~$1$ and $2$, a curve has $N=2(g+1)$  finite branch points, and a double 
branch point at infinity. Let $\nu_i$ be the ramification index of $B_i$. Each branch point is counted $\nu_i - 1$ times.
We assume, that all finite branch points have the ramification index $2$, and 
the branch point $B_0$ at infinity  has $\nu_0 = 3$.

Cases $1$ and $2$ with $\mathcal{T}(x) \equiv 0$ represent $(n,s)$-curves,
which serve as canonical forms of trigonal curves.
In these cases, the genus is  computed by the formula $g=\frac{1}{2}(n-1)(s-1)$,
see \cite{bel99}. 
Case~$3$ is obtained by a proper bi-rational transformation from Case $2$,
thus the both cases have the same Weierstrass gap sequence~$\mathfrak{W}$.
Therefore, classification of trigonal curves is based on Cases $1$ and $2$.   
The corresponding gap sequences are (each set is supposed to be ordered ascendingly)
\begin{align*}
&\text{Case 1}& &\mathfrak{W}=\{3i-2 \mid i=1,\dots, \mFr\} \cup \{3i-1 \mid i=1,\dots, 2\mFr\};& \\
&\text{Case 2}& &\mathfrak{W}=\{3i-1 \mid i=1,\dots, \mFr\} \cup \{3i-2 \mid i=1,\dots, 2\mFr+1\}.&
\end{align*}

The Sato weight introduces an order relation  in the space of monomials $y^j x^i$.
These monomials  are used to construct the equation of a curve, 
and  differentials. 
The ordered lists of monomials in the both cases  are 
\begin{subequations}\label{Monoms}
\begin{align}\label{Monom1}
&\text{Case 1:}& &\mathfrak{M} = \big\{1,  x, \dots, x^{\mFr-1}, x^\mFr, y, x^{\mFr+1}, y x, \dots, 
x^{2\mFr-1}, y x^{\mFr-1}, x^{2\mFr},  \\ 
&& &\quad\quad y x^\mFr,  y^2,  x^{2\mFr+1},  y x^{\mFr+1}, 
\{ y^2 x^i, x^{2\mFr+1+i},  y x^{\mFr+1+i} \mid i\in\Natural \} \big\}, \notag \\  \label{Monom2}
&\text{Case 2:} & &\mathfrak{M} = \big\{ 1,  x,  \dots, 
x^{\mFr-1},  x^\mFr,  y,  x^{\mFr+1}, y x,  \dots, x^{2\mFr-1}, y x^{\mFr-1}, x^{2\mFr},  \\ 
&& & y x^\mFr,  x^{2\mFr+1}, y^2,  yx^{\mFr+1}, x^{2\mFr+2},  
\{y^2 x^i,  yx^{\mFr+1+i}, x^{2\mFr+2+i} \mid i\in\Natural \} \big\}. \notag
\end{align}
\end{subequations}

\subsection{Riemann surface}\label{ss:RiemSurfTC}
In what follows, we focus on the canonical forms of trigonal curves
\begin{gather}\label{TrigCPQEq}
0 = f(x,y)= - y^3 + y\mathcal{Q}(x) + \mathcal{P}(x).
\end{gather}
with the discriminant polynomial
\begin{gather*}
\Delta(x) =  \mathcal{P}(x)^2 - \tfrac{4}{27} \mathcal{Q}(x)^3.
\end{gather*}
Suppose, that all roots of $\Delta$ are distinct, also $\mathcal{P}$ and $\mathcal{Q}$ have no common roots.
These conditions are sufficient for having only branch points
with the ramification index $2$.

Solutions of \eqref{TrigCPQEq} are given by  Cardano's formula:
\begin{subequations}\label{TrigSol}
\begin{align}\label{TrigSolP}
&y_{+,a}(x) =  q_{+,a} (x) + \tfrac{1}{3} \mathcal{Q}(x) q^{-1}_{+,a} (x),\qquad  a=1,\,2,\,3,\\
&q_{+,a} (x) = \upsilon^{1/3} _{+} (x)\rme^{2 (a-1) \imath \pi /3}, \qquad
\upsilon_{+} (x) = \tfrac{1}{2} \Big(\mathcal{P}(x) + \sqrt{\Delta(x)} \Big), \notag
\intertext{or equivalently}
&y_{-,a}(x) =  q_{-,a} (x) + \tfrac{1}{3} \mathcal{Q}(x) q^{-1}_{-,a} (x),\qquad  a=1,\,2,\,3, \label{TrigSolM}\\
&q_{-,a} (x) = \upsilon^{1/3} _{-} (x)\rme^{2 (a-1) \imath \pi /3}, \qquad
\upsilon_{-} (x) = \tfrac{1}{2} \big(\mathcal{P}(x) - \sqrt{\Delta(x)} \big), \notag
\end{align} 
\end{subequations}
where $\sqrt{\Delta}$ is defined by \eqref{SqrtDef}. By $q_{\textsf{s},a}$, $\textsf{s}=\pm 1$, $a=1$, $2$, $3$,
three cubic roots of $\upsilon_{\textsf{s}}$ are denoted. Let
\begin{gather}\label{CubicRootDef}
\upsilon_{\textsf{s}}^{1/3} (x)  = \left\{ \begin{array}{ll}
|\upsilon_{\textsf{s}}(x)|^{1/3} \,\rme^{(\imath/3) \arg  \upsilon_{\textsf{s}}(x)} 
& \text{ if }\quad \arg  \upsilon_{\textsf{s}}(x) \geqslant 0, \\
|\upsilon_{\textsf{s}}(x)|^{1/3}\, \rme^{(\imath/3) \arg \upsilon_{\textsf{s}}(x)+ \imath 2\pi/3} 
& \text{ if }\quad \arg  \upsilon_{\textsf{s}}(x) < 0.
\end{array}  \right.
\end{gather}
According to this definition, the function $\upsilon_{\textsf{s}}^{1/3}$ has the range
$[0,\frac{2}{3}\pi)$, and its discontinuity is located over the contour 
$\Upsilon_{\textsf{s}} = \{x \mid \arg \upsilon_{\textsf{s}}(x) = 0\}$.

 \begin{theo}\label{T:UpsilonDscont}
Let $\upsilon_{\textsf{s}}^{1/3}$ be defined by \eqref{CubicRootDef}. Then each
$y_{\textsf{s},a}$ defined by \eqref{TrigSol} 
 has discontinuity  over the contour 
 $\Upsilon_{\textsf{s}} = \{x \mid \arg \upsilon_{\textsf{s}}(x) = 0\}$.
Along a path from a region with $\arg \upsilon_{\textsf{s}}(x) < 0$ 
to a region with $\arg \upsilon_{\textsf{s}}(x) \geqslant 0$ the analytic continuation of
$y_{\textsf{s},a}$ is given by $y_{\textsf{s},b}$, where $a \mapsto b$ according to one of
 cyclic permutations  $(123)$ or $(132)$.
\end{theo}
\begin{proof}
Let $\textsf{s} \,{=}\,{+}1$, and
 $\tilde{x}$ be located in the vicinity of the contour $\Upsilon_+$, more precisely 
$|\arg \upsilon_+(\tilde{x})| < 3 \phi$, with a small positive value~$\phi$. 
Then $0 \leqslant \arg \upsilon_+^{1/3}(\tilde{x}) < \phi$  if $\arg \upsilon_+(\tilde{x}) \geqslant 0$,
and $\frac{2}{3}\pi - \phi < \arg \upsilon_+^{1/3}(\tilde{x}) < \frac{2}{3}\pi$ if $\upsilon_+(\tilde{x}) < 0$.
Therefore, $\Upsilon_+$ is the contour of discontinuity of  $\upsilon_+^{1/3}$.

Now, we find how the three values $q_{+,a}$ of $\upsilon_+^{1/3}$ are connected over $\Upsilon_+$.
Let $U(x_0; \delta)$ be a disc of radius $\delta$ with the center at $x_0 \in \Upsilon_+$.
The contour $\Upsilon_+$ divides the disc into two parts: $U_+$ where $\arg \upsilon_+(x) \,{\geqslant}\, 0$,  and  
$U_-$ where $\arg \upsilon_+(x) \,{<}\, 0$. 
Let $3 \phi_+ \,{=}\, \max_{\tilde{x} \in U_+} \arg \upsilon_+(\tilde{x})$, 
then the range of $\arg \upsilon_+^{1/3}$ on $U_+$ is $[0,\phi_+)$.
Let ${-}3 \phi_- \,{=}\, \min_{\tilde{x} \in U_-} \arg \upsilon_+(\tilde{x})$,
then the range of $\arg \upsilon_+^{1/3}$ on $U_-$ is $(\frac{2}{3}\pi - \phi_-,\frac{2}{3}\pi)$.
The analytic continuation of $\upsilon_+^{1/3}(x) = q_{+,1}(x)$ from $U_+$ to  $U_-$ is given by
 $\rme^{-\frac{2}{3} \imath \pi}  \upsilon_+^{1/3}(x) = q_{+,3}(x)$, since the range of
 $\arg q_{+,3}$ on $U_-$ is $(-\phi_-,0)$, and so continuously connects to $q_{+,1}$ with
 the range $[0,\phi_+)$ on $U_+$.
Similarly, the analytic continuation of $\upsilon_+^{1/3}(x) = q_{+,1}(x)$
from $U_-$ to  $U_+$ is given by  $\rme^{\frac{2}{3} \imath \pi}  \upsilon_+^{1/3}(x) = q_{+,2}(x)$. 
Therefore, along a path  from $U_-$ to  $U_+$,
the function $q_{+,1}$ continuously connects to $q_{+,2}$,  then $q_{+,2}$ connects to $q_{+,3}$, and
$q_{+,3}$ connects to $q_{+,1}$.

The same is true for $\textsf{s} =-1$.
\end{proof}

 \begin{theo}\label{T:CRDeltaDscont}
Let $\sqrt{\Delta}$ be defined by \eqref{SqrtDef}. Then among three values of $y$, given by \eqref{TrigSolP}, or
\eqref{TrigSolM},
two have  discontinuity over the contour  $\Gamma = \{x \mid \arg \Delta(x) = 0\}$.
 If $\Gamma_i$ is a segment of $\Gamma$ which starts at a branch point $B_i = (e_i,h_i)$, 
 where $h_i = y_a(e_i) = y_b(e_i)$, 
 then $y_a$, $y_b$ are  discontinuous over $\Gamma_i$, and
$y_a$ serves as the analytic continuation of $y_b$ on the other side of $\Gamma_i$, and vice versa. 
\end{theo}
\begin{proof}
According to  Theorem~\ref{T:DeltaDscont}, the both functions 
$\upsilon_+$, $\upsilon_-$, defined in \eqref{TrigSol},
have discontinuity over the contour $\Gamma = \{x \mid \arg \Delta(x) = 0\}$,
and serve as analytic continuations of each other. This implies that all $q_{\textsf{s},a}$ have discontinuity over $\Gamma$,
and  for every $a$ the analytic continuation of $q_{+,a}$ is given by $q_{-,a}$, and vice versa.
We assume that $y_{+,a}$, $a=1$, $2$, $3$, have discontinuity over $\Gamma$,
which follows from the same property of $q_{+,a}$.

Recall the relation 
$\upsilon^{1/3}_+(x) \upsilon^{1/3}_-(x) = \frac{1}{3} \rme^{2 n \imath \pi /3} \mathcal{Q}(x)$,
where $n=0$, $1$, or $2$ such that $\arg \upsilon_+(x) + \arg \upsilon_-(x) = 3 \arg \mathcal{Q}(x) + 2\pi n$ holds
in the vicinity of~$x$. This implies  three equalities of the form $y_{+,a_1}(x) = y_{-,a_2}(x)$,
where $a_1 \mapsto a_2$ according to one of the transpositions: $(12)$, $(13)$, or $(23)$.
Indeed, 

\begin{tabular}{llll}
$n=0$&  $y_{+,1}(x) = y_{-,1}(x)$, & $y_{+,2}(x) = y_{-,3}(x)$, & $y_{+,3}(x) = y_{-,2}(x)$; \\
$n=1$&  $y_{+,1}(x) = y_{-,3}(x)$, & $y_{+,2}(x) = y_{-,2}(x)$, & $y_{+,3}(x) = y_{-,1}(x)$; \\
$n=2$&  $y_{+,1}(x) = y_{-,2}(x)$, & $y_{+,2}(x) = y_{-,1}(x)$, & $y_{+,3}(x) = y_{-,3}(x)$.
\end{tabular}

\noindent
Taking into account, that over $\Gamma$ 
at every $a$  the analytic continuation of $y_{+,a}$ is given by $y_{-,a}$,
we see that among the three values of $y$ given by \eqref{TrigSolP}, or \eqref{TrigSolM}, 
one remains continuous, and 
the other two serve as  analytic continuations of each other.
Indeed, let  $n=0$ in the vicinity of $x$, then $y_{+,1}(x) = y_{-,1}(x)$. 
If a segment of $\Gamma$ is located in the vicinity of $x$, then $y_{-,1}$ serves as the
analytic continuation of $y_{+,1}$ on the other side of  the segment, and so $y_{+,1}$ remains continuous.
At the same time,  $y_{-,2}(x) = y_{+,3}(x)$ serves as the analytic continuation of 
$y_{+,2}(x) = y_{-,3}(x)$, 
as follows from Theorem~\ref{T:DeltaDscont}.

The contour $\Gamma$ consists of segments $\Gamma_i$, each starts at $e_i$ such that $B_i = (e_i,h_i)$
is a branch point, and ends at infinity. Let $x_0 \in \Gamma_i$ be located in the  vicinity of $e_i$
which does not contain $\Upsilon_+$, and  $h_i = y_{+,a}(e_i) = y_{+,b}(e_i)$.
Let $U(x_0,\delta)$ be a disc of radius $\delta$  centered at $x_0$, such that $|x_0-e_i |\geqslant \delta$.
Then $U(x_0,\delta)$ is divided by $\Gamma_i$ into two parts: $U_+$ where $\arg \Delta(x) \geqslant 0$,
and $U_-$ where $\arg \Delta(x) < 0$. There exists such $c$ that $y_{+,c}(x) = y_{-,c}(x)$ for every $x\in U(x_0,\delta)$,
and so $y_{+,c}$ is continuous over $U(x_0,\delta)$. Then $a$, $b$  are the other two values from $\{1,2,3\}$,
since  $y_{+,a}(x) = y_{-,b}(x)$ and  $y_{+,b}(x) = y_{-,a}(x)$ over $U(x_0,\delta)$. Indeed, due to $\Delta(e_i)=0$
we have $y_{+,a}(e_i) = y_{+,b}(e_i)$.

Similar considerations can be made in the case of $\textsf{s}=-1$.
\end{proof}

In what follows, we work with solutions $y_{+,a}$, $a=1$, $2$, $3$, computed by the formula \eqref{TrigSolP},
and denote them by $y_{a}$.

\subsection{Monodromy path}\label{ss:PathTC}
In order to construct a model of the Riemann surface, a monodromy path $\gamma$, 
which is a continuous path through all branch points, should be chosen.
Based on the investigation presented in subsection~\ref{ss:RiemSurfTC}, 
the following algorithm is suggested.
\begin{enumerate}
\renewcommand{\labelenumi}{\arabic{enumi}.}
\item  Find all finite branch points $\{B_i=(e_i,h_i)\}_{i=1}^N$,
and sort ascendingly first by $\Ren e_i$, then by $\Imn e_i$.
The points are enumerated according to this order.
Each point $B_i$ is labeled by '$a$--$b$', such that $h_i = y_{a}(e_i) = y_{b}(e_i)$,
which means that  solutions $y_{a}$ and $y_{b}$ connect over the
segment $\Gamma_i$, which goes from $e_i$ to infinity.

\item
According to this order, a continuous path $\gamma$ on the Riemann sphere through all $e_i$ is constructed from 
 straight line segments $[e_i,e_{i+1}]$, $i=1$, \ldots, $N$.
Then the segment $(-\infty, e_1]$ is added at the beginning
of the polygonal path~$\gamma$, and $[e_N,\infty)$ at the end.

\item Plot the contour $\Gamma = \{x \mid \arg \Delta(x) = 0\}$, see blue contours on fig.\,\ref{f:C34Path},
and the contour $\Upsilon_+ = \{x \mid \arg \upsilon_+(x) = 0\}$, see green contours on fig.\,\ref{f:C34Path}. 
Identify a permutation which corresponds to each segment of $\Gamma_i$, and $\Upsilon_+$.
When the path $\gamma$ crosses $\Gamma_i$, solutions $y_{a}$ and $y_{b}$ interchange.
When $\gamma$ crosses $\Upsilon_+$, solutions $y_{a}$, $y_{b}$, $y_{c}$ permute in the 
corresponding cycle.

\item
Moving along the monodromy path $\gamma$ from the left to the right, passing $e_i$ and cuts from the below,
counter-clockwise, and taking into account
intersections with  $\Gamma$ and $\Upsilon_+$, find the sequence of changes 
of solutions $y_{a}$, $a=1$, $2$, $3$, 
starting from $1$, $2$, and $3$, respectively:
$$
 \begin{array}{ll}
\text{Sheet \textsf{a}: } & 
\{\textsf{a}_{0,1}=1\}\cup \{\textsf{a}_{i,i+1}\}_{i=1}^{N-1} \cup \{\textsf{a}_{N,0}\}, \\
\text{Sheet \textsf{b}: } & 
\{\textsf{b}_{0,1}=2\}\cup \{\textsf{b}_{i,i+1}\}_{i=1}^{N-1}\cup \{\textsf{b}_{N,0}\}, \\
\text{Sheet \textsf{c}: } & 
\{\textsf{c}_{0,1}=3\}\cup \{\textsf{c}_{i,i+1}\}_{i=1}^{N-1} \cup \{\textsf{c}_{N,0}\}. \\
\end{array}
$$
Index $0$ stands for infinity. 
Each sequence marks a sheet of the Riemann surface.
\end{enumerate}

The monodromy path $\gamma$ lifted to each sheet, and closed by a semi-circle around infinity in the
counter-clockwise direction, is homotopic to zero.

\subsection{Homology}\label{ss:HomTC}

Note, that $N=2(g+1)$ on a canonical trigonal curve.
Therefore, finite branch points split into $g+1$ pairs. 
Without loss of generality, $g$ cuts are made 
between branch points in pairs  $B_{2i}, B_{2i+1}$, $i=1$, \ldots, $g$.
Each cut goes through finite points, on the two sheets connected by branch points in a pair. 
One more cut is made from $B_1$ to $B_N$ through infinity.

Let $\mathfrak{a}_i$ encircle the  cut $(B_{2i}, B_{2i+1})$ counter-clockwise.
Let $\mathfrak{b}_i$ emerge  from the cut $(B_1, \infty)\cup (\infty,B_N)$, 
and enter the cut encircled by~$\mathfrak{a}_i$. 
In this way a canonical homology basis is obtained.

\subsection{Cohomology}\label{ss:coHomTC}

First kind differentials are constructed with the help of the first $g$ monomials 
from the ordered list $\mathfrak{M}$, namely:
\begin{align}\label{Dif1TC}
\rmd u_{\mathfrak{w}_i} = \frac{\mathfrak{m}_{g-i+1}(x,y) \, \rmd x}{\partial_y f(x,y)},\quad
i = 1,\dots, g,
\end{align}
where $\mathfrak{m}_{k}$ is the $k$-th element of $\mathfrak{M}$,
and  $\mathfrak{w}_i$ are elements of the  gap sequence $\mathfrak{W}$.

A second kind differential $\rmd r_{\mathfrak{w}_i}$
is constructed  with the help of the first $g + i$ monomials from $\mathfrak{M}$. Namely
\begin{align}\label{Dif2TC}
\rmd r_{\mathfrak{w}_i} = \Big(\sum_{j=1}^{g+i} c_{i,j} \mathfrak{m}_{j}(x,y) \Big) \frac{\rmd x}{\partial_y f(x,y)},\quad
i = 1,\dots, g.
\end{align}
The relation \eqref{urRel} defines the coefficients of monomials $\mathfrak{m}_{g+i}$, $i\geqslant 1$.
Coefficients of the remaining part of the sum in \eqref{Dif2TC} are also essential. 
In the case of trigonal curves, the second kind differentials associated to 
the standard first kind differentials \eqref{Dif1TC}
are obtained in \cite{bel00}, by means of the Klein formula.

\subsection{Computation of periods}
Next, first kind integrals on each segment along the polygonal monodromy path $\gamma$ are computed:
\begin{subequations}\label{AintTC}
\begin{align}
&\mathcal{A}_{i,i+1}^{[\textsf{n}_{i,i+1}]} 
= \int_{e_{i}}^{e_{i+1}} \rmd u^{[\textsf{n}_{i,i+1}]} ,\qquad i=1,\, \dots,\, N-1,\\
&\mathcal{A}^{[\textsf{n}_{0,1}]}_{0,1} = \int_{-\infty}^{e_1} \rmd u^{[\textsf{n}_{0,1}]},\qquad \qquad
\mathcal{A}^{[\textsf{n}_{N,0}]}_{N,0} = \int_{e_N}^{\infty} \rmd u^{[\textsf{n}_{N,0}]},
\end{align}
\end{subequations} 
where $\textsf{n} = \textsf{a}$, $ \textsf{b}$, or $ \textsf{c}$, depending on the sheet.
The integrand of $\mathcal{A}_{i,j}^{[\textsf{n}_{i,j}]} $ 
is defined by \eqref{Dif1TC}  with $y = y_{\textsf{n}_{i,j}}(x)$.
A lift of the monodromy path $\gamma$ to each sheet, closed by a semi-circle around infinity
in the counter-clockwise direction,
is homotopic to zero, an so  relations hold
\begin{gather}\label{TrigInvRels}
\mathcal{A}^{[\textsf{n}_{0,1}]}_{0,1} + 
\sum_{i=1}^{N-1} \mathcal{A}_{i,i+1}^{[\textsf{n}_{i,i+1}]}  
 + \mathcal{A}^{[\textsf{n}_{N,0}]}_{N,0}  = 0, \qquad
 \textsf{n} = \textsf{a},\, \textsf{b},\, \textsf{c},
\end{gather}
which serve for verification. 

With a chosen basis of canonical cycles,  first kind period matrices $\omega$, and $\omega'$
are computed by \eqref{omegaM}, and the normalized period matrix is obtained by
$$\tau = \omega^{-1} \omega'.$$
The latter is required  to be symmetric with a positive imaginary part.

With the same basis of canonical cycles,  second kind period matrices $\eta$, and $\eta'$
are computed by \eqref{etaM} with differentials \eqref{Dif2TC}.
Then the symmetric matrix $\varkappa$ from the definition of
$\sigma$-function is obtained by
$$\varkappa =  \eta \omega^{-1}.$$
First and second kind period matrices satisfy the Legendre relation \eqref{LegRel}.

\subsection{Example 3: $(3,4)$-curve}
We consider the simplest trigonal curve in its canonical form
\begin{gather}\label{C34Eq}
0 = f(x,y)= - y^3 + x^4 + \lambda_2y x^2 + \lambda_{3} x^3 + \lambda_{5} y x   + \lambda_{6} x^2 +
\lambda_{8} y + \lambda_{9} x + \lambda_{12}.
\end{gather}
From $\Delta(x)=0$ we find $x$-coordinates $e_i$ of branch points $B_i = (e_i,h_i)$. Then using \eqref{TrigSolP},
we find the corresponding values $y_{a}(e_i)$, $a=1$, $2$, $3$, two of which, say $a$ and $b$, coincide and give $h_i$.
So each branch point is labeled by `$a\text{--}b$', that indicates which  solutions connect in the vicinity of $B_i$.

A $(3,4)$-curve possesses the gap sequence $\mathfrak{W} = \{1,\,2,\,5\}$,
and the first kind differentials have the form
$$ \rmd u = \begin{pmatrix} \rmd u_1 \\ \rmd u_2 \\ \rmd u_5 \end{pmatrix} 
= \begin{pmatrix} y \\ x \\ 1  \end{pmatrix} 
\frac{\rmd x}{-3 y^2+\mathcal{Q}(x)}. $$
Let $\mathcal{A}^{[a]}_{i,j}$ denote a first kind integral between  $e_i$ and $e_j$
computed with   $y = y_a(x)$:
\begin{gather}
\mathcal{A}^{[a]}_{i,j} = \int_{e_i}^{e_j} \rmd u^{[a]}.
\end{gather} 
Second kind differentials associated with the first ones on the curve \eqref{C34Eq} are defined as follows
$$ \rmd r = \begin{pmatrix} \rmd r_1 \\ \rmd r_2 \\ \rmd r_5 \end{pmatrix} 
= \begin{pmatrix} x^2 \\ 2 x y \\  R_{5} \end{pmatrix} 
\frac{\rmd x}{-3 y^2+\mathcal{Q}(x)}, $$
\begin{align*}
&R_{5} = 5x^2 y + 3 \lambda_3 x y  + \tfrac{2}{3}\lambda_2^2 x^2 
+ \lambda_6 y + \tfrac{2}{3} \lambda_2 \lambda_5 x.
 \end{align*}
A second kind integral $\mathcal{B}^{[a]}_{i,j}$ between $e_i$ and $e_j$ with  $y = y_a(x)$
is computed by
\begin{gather}
\mathcal{B}^{[a]}_{i,j} = \int_{B_i}^{B_j} \rmd r^{[a]}.
\end{gather} 

\bigskip
As an example, we consider a curve defined by the equation
\begin{gather}\label{C34Ex}
0 = f(x,y) \equiv -y^3 + x^4 + y \big(4x^2 + 5x +11) + 3x^3 + 7 x^2 + 16 x + 9.
\end{gather}
The curve has eight finite branch points:
\begin{align*}
&e_1 \approx -4.58931,& &h_1 \approx -4.9092& & 2\text{--}3& \\
&e_2 \approx -1.17922 - 0.934455 \imath,& &h_2 \approx 1.60505 + 0.430221 \imath& &1\text{--}3& \\
&e_3 \approx -1.17922 + 0.934455 \imath,& &h_3 \approx 1.60505 - 0.430221 \imath& &1\text{--}3& \\
&e_4 \approx -0.431732 - 2.20256 \imath,& &h_4 \approx 0.309255 - 1.83532 \imath& &2\text{--}3&\\
&e_5 \approx -0.431732 + 2.20256 \imath,& &h_5 \approx 0.309255 + 1.83532 \imath& &1\text{--}2&\\
&e_6 \approx 0.499118 - 1.57527 \imath,& &h_6 \approx -1.80047 + 1.31135  \imath& &1\text{--}2& \\
&e_7 \approx 0.499118 + 1.57527 \imath,& &h_7 \approx -1.80047 - 1.31135  \imath& &2\text{--}3&  \\
&e_8 \approx 0.812986,& &h_8 \approx -2.42959& &2\text{--}3&
\end{align*}
In the last column the corresponding pairs `$a\text{--}b$' are indicated.

Fig.~\ref{f:C34Path}(a) displays positions of  $e_i$,
 and the contours where solutions $y_a$, $a=1$,\,$2$,\,$3$, have discontinuity.
\begin{figure}[h]
\includegraphics[width=0.32\textwidth]{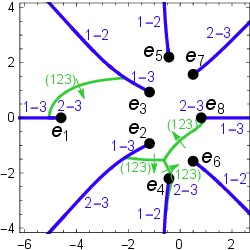}\ 
\includegraphics[width=0.32\textwidth]{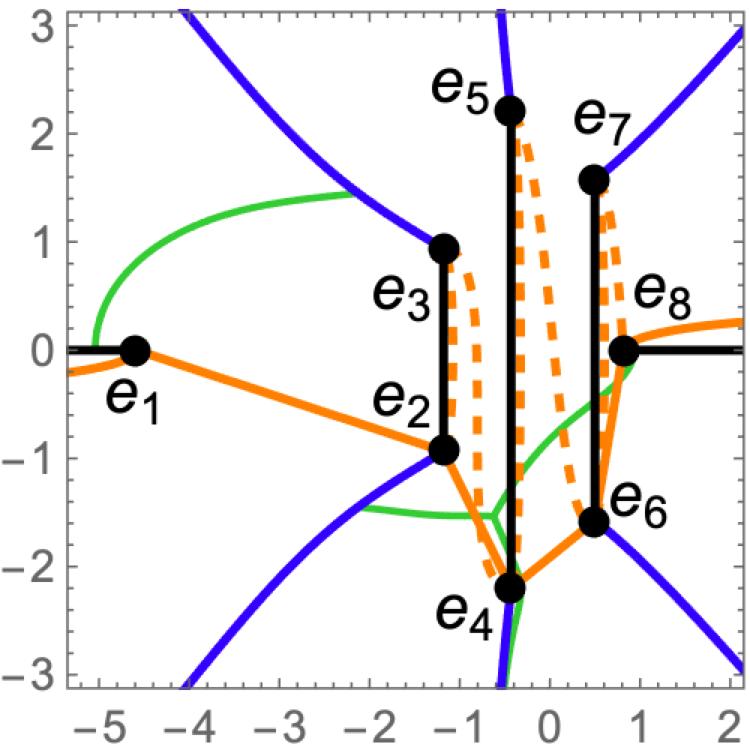}\ 
\includegraphics[width=0.32\textwidth]{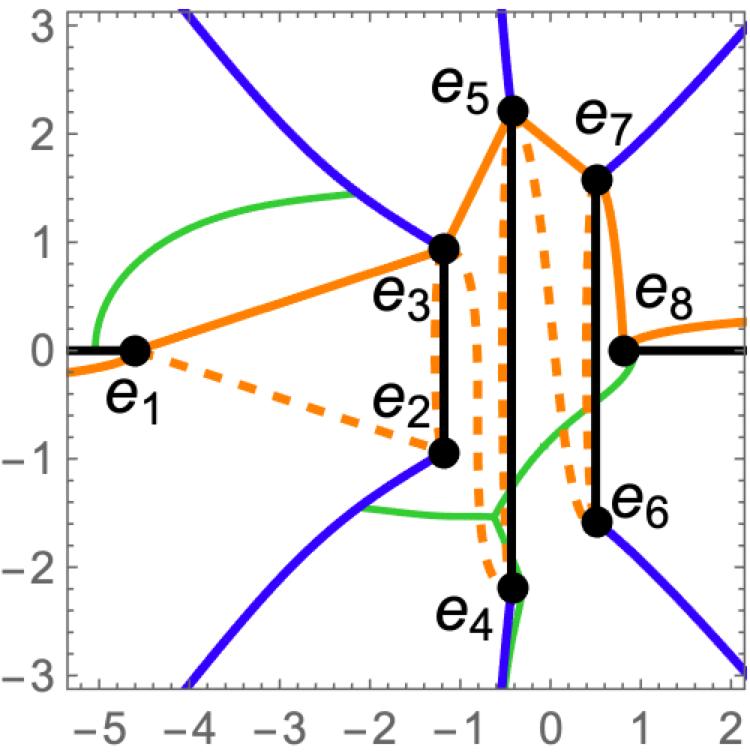}\\
\parbox[b]{0.32\textwidth}{\centering (a) Permutations of \\ $\quad$ solutions
$y_a$, $a=1$,\,$2$,\,$3$}\
\parbox[b]{0.32\textwidth}{\centering (b) Paths $\gamma_{\text{L}}$ (dashed) \\
and $\widetilde{\gamma}_{\text{L}}$ (solid)}\
\parbox[b]{0.32\textwidth}{\centering (c)  Paths $\gamma_{\text{U}}$  (dashed) \\
and $\widetilde{\gamma}_{\text{U}}$ (solid)} 
\caption{Contours $\Gamma$ (blue), $\Upsilon_+$ (green), cuts (black), and
 a path (orange).}\label{f:C34Path}
\end{figure}
The contour $\Upsilon_+ = \{x \mid \arg \upsilon_+(x) = 0\}$ is
marked in green. Over  $\Upsilon_+$ all three solutions connect in pairs.
Each segment of $\Upsilon_+$ is labeled by the cyclic permutation $(123)$,
and an arrow shows in which direction this permutation occurs. 

The contour $\Gamma = \{x \mid \arg \Delta(x) = 0\}$ is
marked in blue. 
Each segment $\Gamma_i$ of $\Gamma$ starts at $e_i$
and ends at infinity. Let $h_i = y_{a}(e_i) =  y_{b}(e_i)$,
then $\Gamma_i$ is labeled by `$a\text{--}b$' in the vicinity of $e_i$.
If $\Upsilon_+$ intersects $\Gamma_i$ at a point $d_i$, then the segment of  $\Gamma_i$
between $d_i$ and infinity  
is labeled by `$b\text{--}c$' such that $y_{b}(x) =  y_{c}(x)$ for all $x$ on this segment.
Fig.\,\ref{f:C34Path}(a) is in accordance  with the
density plots of $\arg y_{a}$, $a=1$, $2$, $3$, shown on fig.\,\ref{f:C34Dens}.
\begin{figure}[h]
\parbox[b]{0.305\textwidth}{\includegraphics[width=0.305\textwidth]{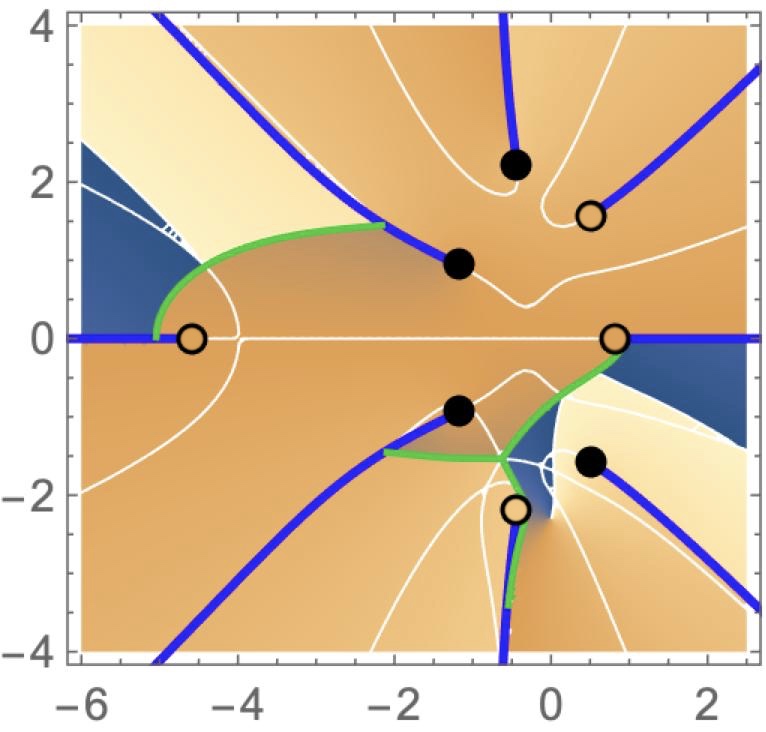}}
\parbox[b]{0.305\textwidth}{\includegraphics[width=0.305\textwidth]{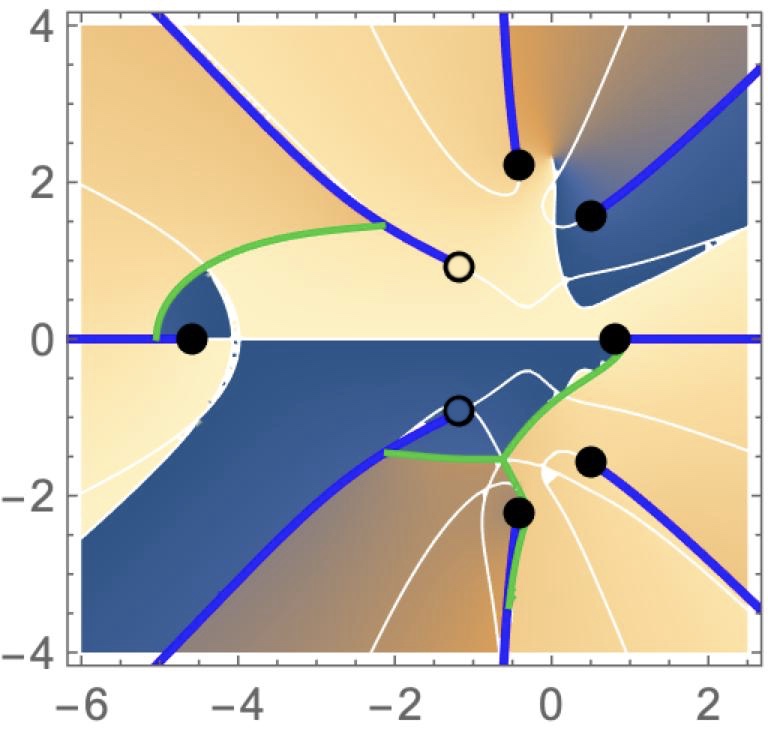}}
\parbox[b]{0.305\textwidth}{\includegraphics[width=0.305\textwidth]{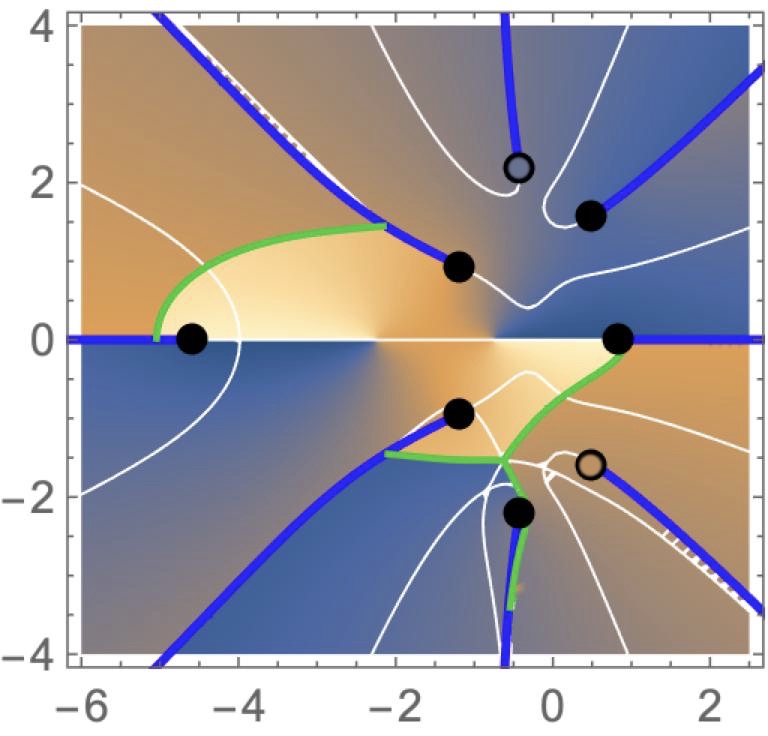}}
\parbox[b]{0.05\textwidth}{\includegraphics[width=0.05\textwidth]{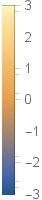} \\
} \\
$\phantom{mmm}  a=1\phantom{mmmmmmmmmmm} a=2 \phantom{mmmmmmmmmmm}  a=3 
\phantom{mmm} $
\caption{Density plot of $\arg y_{a}$,  contours $\Gamma$ (blue), 
$\Upsilon_+$ (green).}\label{f:C34Dens}
\end{figure}

Cuts connect pairs of points: $B_2$ and $B_3$ on  sheets with solutions $y_{1}$ and $y_{3}$;
$B_4$ and $B_5$ on  sheets with  solutions $y_{1}$ and $y_{2}$ in the vicinity of $B_5$
which change consequently into the pair $1$--$3$ and then $2$--$3$ when approaching $B_4$;
 $B_6$ and $B_7$ on sheets with  solutions $y_{2}$ and $y_{3}$ in the vicinity
of $B_7$ which change into the pair $1$--$2$ when approaching $B_6$.
One more cut goes from  $B_8$ to infinity, and then to $B_1$.
In the vicinity of $B_8$  solutions $y_2$ and $y_3$ connect, which change into
the pair $1$--$3$ in the vicinity of infinity, and then into $2$--$3$
in the vicinity of $B_1$.

Let the monodromy path $\gamma_{\text{L}}$ through all branch points be 
\begin{gather*}
(-\infty,e_1] \cup [e_1,e_2] \cup [e_2,e_3] \cup [e_3,e_4] \cup [e_4,e_5] \\
\phantom{mmmmm}\cup [e_5,e_6] \cup [e_6,e_7] \cup [e_7,e_8]
\cup [e_8,\infty),
\end{gather*}
as shown on fig.\,\ref{f:C34Path}(b) in orange. 
The path is polygonal, and goes below points $e_i$ and cuts, 
and so rounds cuts and branch points counter-clockwise.
The path $\gamma_{\text{L}}$ is marked by dashed lines between $e_2$ and $e_8$, and
continuously deformed into a simpler path $\widetilde{\gamma}_{\text{L}}$,
which is marked by solid lines, and goes below points $e_2$, $e_4$, $e_6$.

Sheets are marked according to the monodromy path $\gamma_{\text{L}}$, namely 
\begin{equation}\label{LPathSeq}
 \begin{array}{p{0.9cm}cccccccc p{0.9cm}}
&\ \{\textsf{a}_{0,1}, & \textsf{a}_{1,2}, & \textsf{a}_{2,3}, &
\textsf{a}_{3,4}, & \textsf{a}_{4,5}, &  \textsf{a}_{5,6}, &  \textsf{a}_{6,7}, &
 \textsf{a}_{7,8}, &  $\textsf{a}_{8,0} \} =$\\
\text{Sheet \textsf{a}:} & \{\ \ 1, \ \ \ & 1, & 3,  & 3\text{--}1, & 
1\text{--}2\text{--}3,& 3\text{--}2,& 1\text{--}2, & 2, & $3\ \ \ \},$ \\
&\ \{\textsf{b}_{0,1}, & \textsf{b}_{1,2}, & \textsf{b}_{2,3}, &
\textsf{b}_{3,4}, & \textsf{b}_{4,5}, &  \textsf{b}_{5,6}, &  \textsf{b}_{6,7}, &
 \textsf{b}_{7,8}, & $ \textsf{b}_{8,0} \} =$\\
\text{Sheet \textsf{b}:}  & \{\ \ 2, \ \ \ & 2, & 2, & 2\text{--}3, 
& 2\text{--}3\text{--}1,& 1\text{--}3,& 3\text{--}1, & 1, & $1\ \ \ \},$ \\
&\ \{\textsf{c}_{0,1}, & \textsf{c}_{1,2}, & \textsf{c}_{2,3}, &
\textsf{c}_{3,4}, & \textsf{c}_{4,5}, &  \textsf{c}_{5,6}, &  \textsf{c}_{6,7}, &
 \textsf{c}_{7,8}, &  $\textsf{c}_{8,0} \} =$ \\
\text{Sheet \textsf{c}:} & \{\ \ 3, \ \ \ & 3, & 1, & 1\text{--}2, & 
3\text{--}1\text{--}2,& 2\text{--}1,& 2\text{--}3, & 3, & $2\ \ \ \}.$
\end{array}
\end{equation}
Connection of solutions $y_1$, $y_2$, $y_3$ 
on each sheet is shown on   fig.\,\ref{f:C34Sheets}.
\begin{figure}[h]
\parbox[b]{0.3\textwidth}{\includegraphics[width=0.3\textwidth]{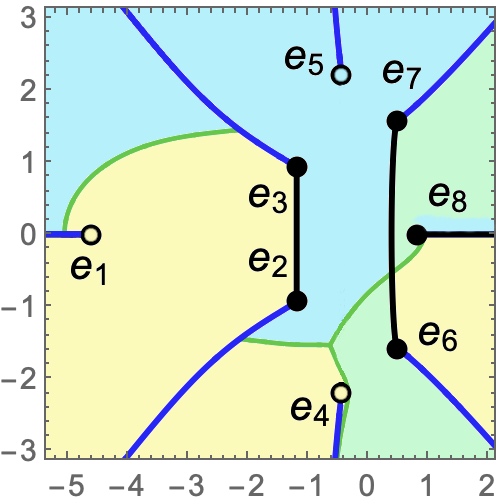}}
\parbox[b]{0.3\textwidth}{\includegraphics[width=0.3\textwidth]{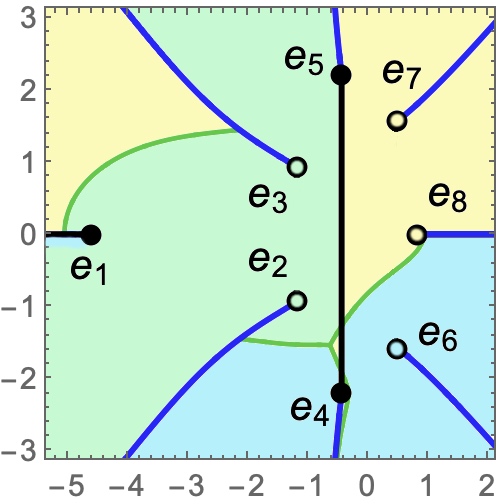}}
\parbox[b]{0.3\textwidth}{\includegraphics[width=0.3\textwidth]{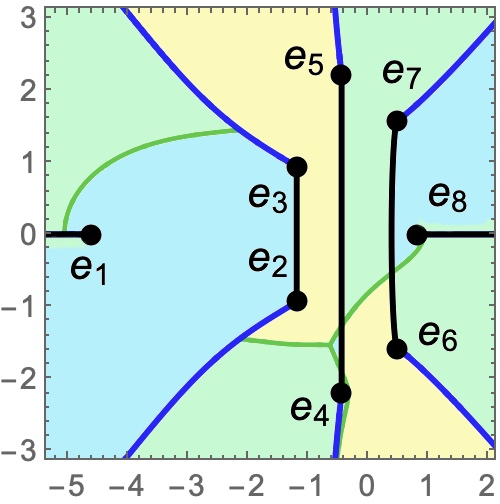}}
\parbox[b]{0.07\textwidth}{\includegraphics[width=0.07\textwidth]{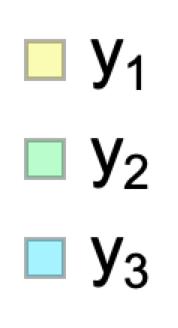}\\
$\quad$\\ $\quad$\\ $\quad$} \\
\parbox[b]{0.3\textwidth}{\centering Sheet \textsf{a}}
\parbox[b]{0.3\textwidth}{\centering Sheet \textsf{b}}
\parbox[b]{0.3\textwidth}{\centering Sheet \textsf{c}} \ \ \
\caption{Connection of solutions $y_1$, $y_2$, $y_3$ on each sheet.}\label{f:C34Sheets}
\end{figure}

On fig.\,\ref{f:C34Path}(c) an alternative monodromy path $\gamma_{\text{U}}$  is presented,
marked by dashed lines between $e_1$ and $e_7$.
The path $\gamma_{\text{U}}$ goes from the right to the left
above points $e_i$ and cuts, which are rounded counter-clockwise.
$\gamma_{\text{U}}$  is continuously deformed into a simpler path $\widetilde{\gamma}_{\text{U}}$,
which is marked by solid lines, and goes above points $e_3$, $e_5$, $e_7$. 

One can see that the cuts $(B_2,B_3)$, and  $(B_6,B_7)$ connect 
Sheets\;\textsf{a} and \textsf{c}, and the cut  $(B_4,B_5)$ connects Sheets\;\textsf{b} and \textsf{c}.
The cut $(B_8,\infty) \cup (\infty,B_1)$ connects Sheets\;\textsf{a} and \textsf{c} on segment $[e_8,\infty)$,
and Sheets\;\textsf{b} and \textsf{c}  on  $(\infty,e_1]$. 

Moving along the path $\widetilde{\gamma}_{\text{L}}$ from the left to the right,
rounding infinity clockwise, and then 
moving along the path $\widetilde{\gamma}_{\text{U}}$ from the right to the left,
a circle $\gamma_{\infty}$ around infinity on the Riemann sphere is drawn. 
Then $\gamma_{\infty}$ is lifted to the curve.

Let the lift of $\gamma_{\infty}$ start at $B_1$ on Sheet\;\textsf{c},  and reach $B_8$ 
along $\tilde{\gamma}_{\text{L}}^{\textsf{c}}$,
the lift of~$\widetilde{\gamma}_{\text{L}}$ to Sheet\;\textsf{c}. 
The analytic continuation of $y_3$ on Sheet\;\textsf{c} is $y_2$ on Sheet\;\textsf{a}
on the other side of $\Gamma_8$ in the vicinity of $e_8$.
So the path enters the cut and emerges on Sheet\;\textsf{a}, 
where segment $[e_8,\infty + \imath \epsilon)$ is governed by $y_2$.
The path continues  on Sheet\;\textsf{a} along~$\tilde{\gamma}_{\text{U}}^{\textsf{a}}$, 
the lift of $\widetilde{\gamma}_{\text{U}}$ to Sheet\;\textsf{a},
and reaches the point $(e_1,y_1(e_1))$, which is not a branch point.
Thus, one full turn around infinity, denoted by $\tilde{\gamma}_\infty^{\textsf{c-a}}$, is completed.
The path continues on Sheet \textsf{a} along  $\tilde{\gamma}_{\text{L}}^{\textsf{a}}$, and reaches $B_8$.
Solution $y_2$ on Sheet \textsf{a} connects to $y_3$ on Sheet \textsf{c} 
on the other side of $\Gamma_8$ in the vicinity of $e_8$. 
So, the path moves from Sheet\;\textsf{a} to Sheet \textsf{c},
where segment $[e_8, \infty + \imath \epsilon)$ is governed by $y_3$.
Then the path goes  on Sheet\;\textsf{c} along  $\tilde{\gamma}_{\text{U}}^{\textsf{c}}$, and reaches~$B_1$. 
At this point the second turn around infinity, denoted by $\tilde{\gamma}_\infty^{\textsf{a-c}}$, is completed.
In the vicinity of $e_1$ solution $y_3$ on Sheet\;\textsf{c} connects to $y_2$ on Sheet\;\textsf{b}.
So the path enters the cut and emerges on Sheet\;\textsf{b}, then goes on Sheet\;\textsf{b} 
along  $\tilde{\gamma}_{\text{L}}^{\textsf{b}}$, and reaches the point $(e_8,y_1(e_8))$, then  
continues on Sheet\;\textsf{b} along $\tilde{\gamma}_{\text{U}}^{\textsf{b}}$, and reaches $B_1$.
The third turn around infinity, denoted $\tilde{\gamma}_\infty^{\textsf{b-b}}$, is completed by the arrival to the initial point.
The lift of $\gamma_{\infty}$, which is $\bar{\gamma}_{3\infty} \equiv \tilde{\gamma}_\infty^{\textsf{c-a}}\cup 
\tilde{\gamma}_\infty^{\textsf{a-c}}\cup \tilde{\gamma}_\infty^{\textsf{b-b}}$, can be used to
reach an arbitrary point on the curve.

In fact, $\textsf{a}_{8,0}$,  $\textsf{c}_{8,0}$ 
in the sequences \eqref{LPathSeq} do not belong to the sheets where they are listed. But
the indicated paths, closed by a semi-circle around infinity in the counter-clockwise direction, 
are homotopic to zero. Correspondingly,  the sum of first kind integrals
along each path vanishes. Then  the following relations are obtained along 
the simpler path $\widetilde{\gamma}_{\text{L}}$:
\begin{gather}\label{C34Rels}
\begin{split}
&\mathcal{A}^{[1]}_{0-,1} + \mathcal{A}^{[1]}_{1,2} 
+ \mathcal{A}^{[3\text{-}1]}_{2,4} + \mathcal{A}^{[1\text{-}2]}_{4,6}
+ \mathcal{A}^{[1\text{-}2]}_{6,8} + \mathcal{A}^{[3]}_{8,0+} = 0,\\
&\mathcal{A}^{[2]}_{0-,1} + \mathcal{A}^{[2]}_{1,2} 
+ \mathcal{A}^{[2\text{-}3]}_{2,4} + \mathcal{A}^{[2\text{-}3]}_{4,6}
+ \mathcal{A}^{[3\text{-}1]}_{6,8} + \mathcal{A}^{[1]}_{8,0+}  = 0,\\
&\mathcal{A}^{[3]}_{0-,1} + \mathcal{A}^{[3]}_{1,2} 
+ \mathcal{A}^{[1\text{-}2]}_{2,4} + \mathcal{A}^{[3\text{-}1]}_{4,6}
+ \mathcal{A}^{[2\text{-}3]}_{6,8} + \mathcal{A}^{[2]}_{8,0+}  = 0.
\end{split}
\end{gather}
Here $0-$ ($0+$) stands for $-\infty - \imath \epsilon$  ($\infty + \imath \epsilon$),
and the superscript of $\mathcal{A}^{[a\text{-}b]}_{i,j}$ indicates that over the segment $[e_i,e_j]$ 
 solution $y_{a}$ changes into $y_{b}$. Note, the last relation follows from the 
first two, due to $\mathcal{A}^{[\textsf{a}_{i,j}]}_{i,j} + \mathcal{A}^{[\textsf{b}_{i,j}]}_{i,j} + \mathcal{A}^{[\textsf{c}_{i,j}]}_{i,j} = 0$,
where $\textsf{a}_{i,j}$, $\textsf{b}_{i,j}$, and $\textsf{c}_{i,j}$ correspond to Sheets $\textsf{a}$, $\textsf{b}$, and $\textsf{c}$.

Let the path $\widetilde{\gamma}_{\text{U}}$, starting at $\infty+\imath \epsilon$ and ending at 
$-\infty-\imath \epsilon$, be closed by a semi-circle around infinity in the counter-clockwise direction.
On each sheet this path is  homotopic to zero,
and so produces another set of relations:
\begin{gather}\label{C34Rels2}
\begin{split}
&\mathcal{A}^{[3]}_{0+,8} + \mathcal{A}^{[2]}_{8,7}  
+ \mathcal{A}^{[3]}_{7,5}  + \mathcal{A}^{[3]}_{5,3}  
+ \mathcal{A}^{[1]}_{3,1}  + \mathcal{A}^{[1]}_{1,0-}  = 0,\\
&\mathcal{A}^{[2]}_{0+,8} + \mathcal{A}^{[3]}_{8,7} 
+ \mathcal{A}^{[2]}_{7,5} + \mathcal{A}^{[1]}_{5,3} 
+ \mathcal{A}^{[3]}_{1,3}  +\mathcal{A}^{[2]}_{1,0-} = 0,\\
& \mathcal{A}^{[1]}_{0+,8} + \mathcal{A}^{[1]}_{8,7}
+ \mathcal{A}^{[1]}_{7,5} + \mathcal{A}^{[2]}_{5,3} 
+ \mathcal{A}^{[2]}_{3,1}  + \mathcal{A}^{[3]}_{1,0-} = 0.
\end{split}
\end{gather}
 
\begin{figure}[h]
\parbox[b]{.56\textwidth}{\includegraphics[width=0.5\textwidth]{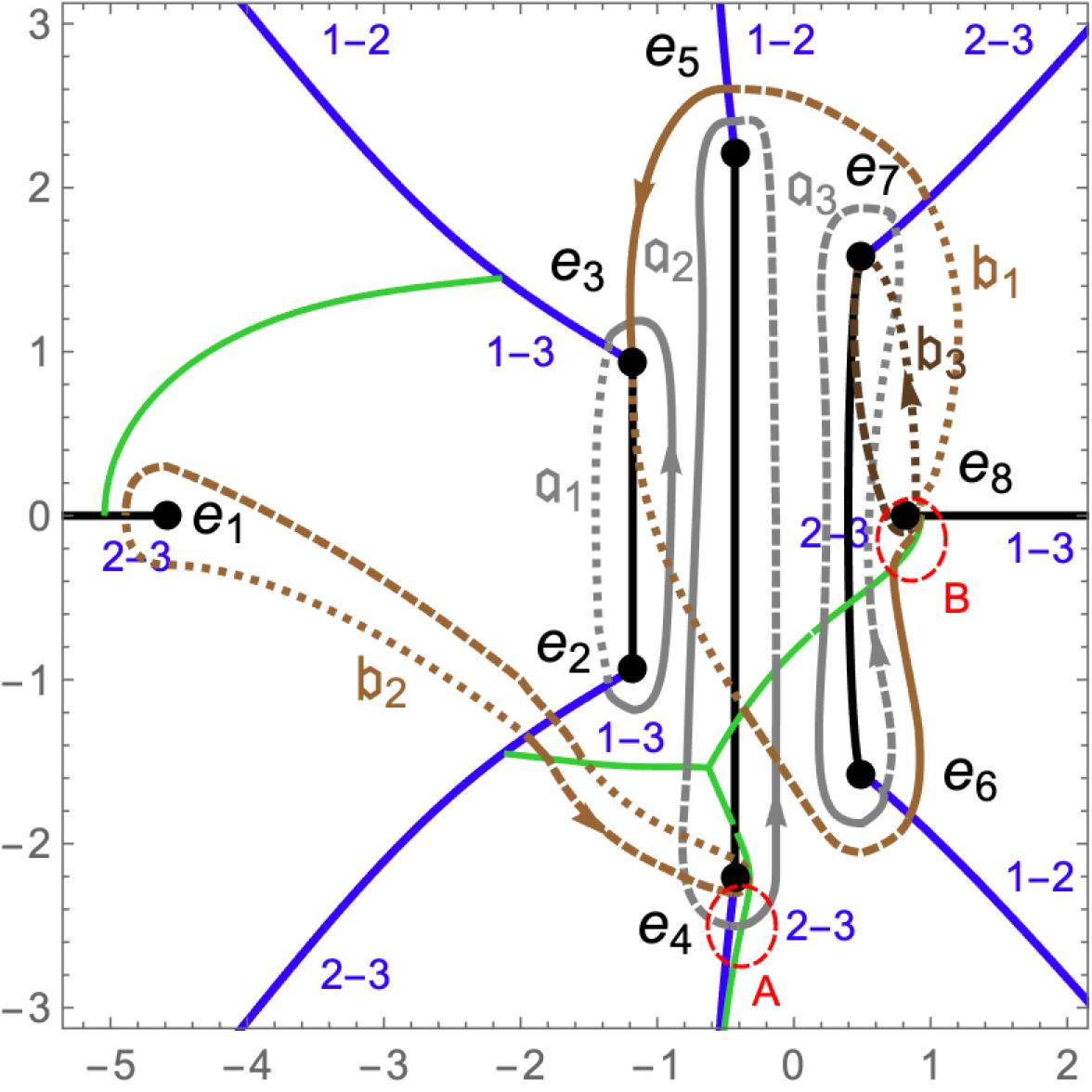}}
\parbox[b]{.16\textwidth}{\includegraphics[width=0.14\textwidth]{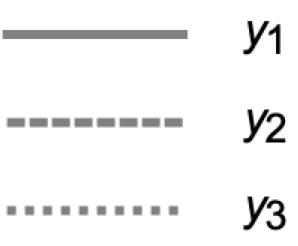} \\
$\quad$ \\ $\quad$ \\ $\quad$ \\ $\quad$  \\
$\quad$ \\ $\quad$ \\ $\quad$ \\ $\quad$ \\ $\quad$ \\ $\quad$ }  $\quad$  $\quad$ \\
\includegraphics[width=0.3\textwidth]{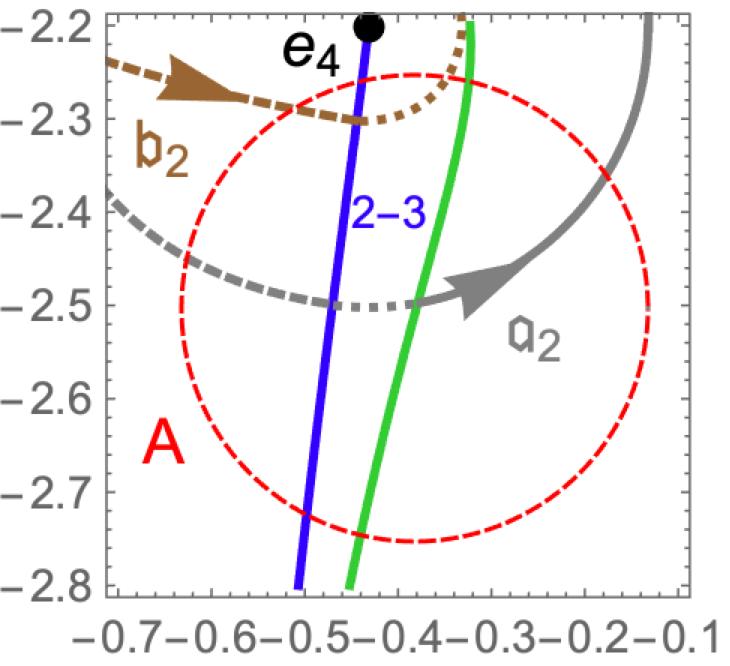}$\qquad$
\includegraphics[width=0.3\textwidth]{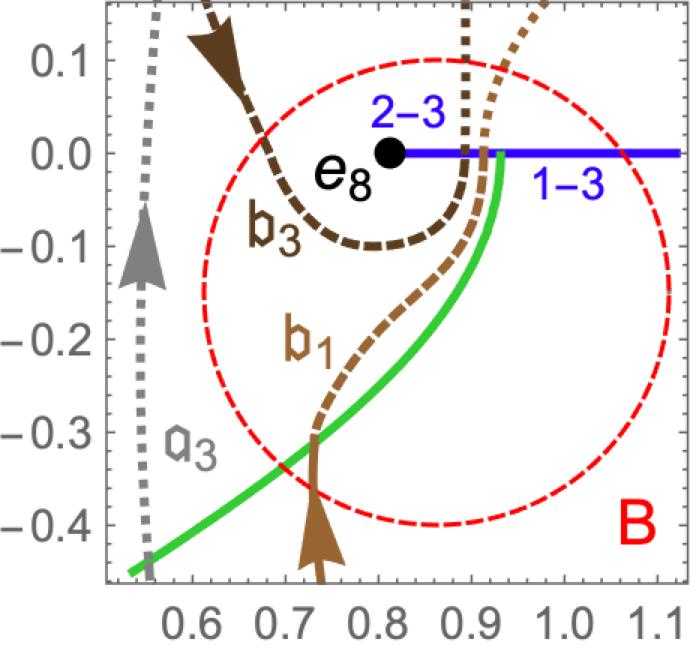} $\qquad\qquad$
\caption{Canonical homology basis.}\label{f:C34Cycles}
\end{figure}

Note, that all branch points can be reached on Sheet $\textsf{c}$, as well as all cuts. Let
$\mathfrak{a}$-cycles be located on Sheet $\textsf{c}$, encircling the three finite cuts. Let
$\mathfrak{b}_i$-cycle emerge from the cut  $(B_8,\infty)\cup (\infty,B_1)$, and
 enter the cut encircled by $\mathfrak{a}_i$-cycle, see fig.\,\ref{f:C34Cycles}.
Therefore, $\mathfrak{b}_1$ and $\mathfrak{b}_3$ go through Sheets $\textsf{c}$ and $\textsf{a}$, then
$\mathfrak{b}_2$  goes through Sheets $\textsf{c}$ and~$\textsf{b}$.
Periods are calculated according to  fig.\,\ref{f:C34Cycles}, as follows
\begin{align*}
&\omega_1 = \mathcal{A}^{[1]}_{2,3} + \mathcal{A}^{[3]}_{3,2},\\
&\omega_2 = \mathcal{A}^{[3\text{-}1\text{-}2]}_{4,5}  + \mathcal{A}^{[1\text{-}3\text{-}2]}_{5,4}  
= \mathcal{A}^{[3\text{-}1]}_{4,6} + \mathcal{A}^{[1\text{-}2]}_{6,5} 
+  \mathcal{A}^{[1]}_{5,2} + \mathcal{A}^{[1\text{-}2]}_{2,4},\\
&\omega_3 =  \mathcal{A}^{[2\text{-}3]}_{6,7}  +  \mathcal{A}^{[2\text{-}1]}_{7,6},\\
&\omega'_1 = \mathcal{A}^{[3]}_{8,7} + \mathcal{A}^{[2]}_{7,5} + \mathcal{A}^{[1]}_{5,3} 
+ \mathcal{A}^{[3\text{-}2]}_{3,6}  + \mathcal{A}^{[1\text{-}2]}_{6,8}, \\
&\omega'_2 = \mathcal{A}^{[3]}_{1,2} + \mathcal{A}^{[1\text{-}2]}_{2,4}
+\mathcal{A}^{[3\text{-}2]}_{4,2} + \mathcal{A}^{[2]}_{2,1}, \\
&\omega'_3 =  \mathcal{A}^{[3]}_{8,7} + \mathcal{A}^{[2]}_{7,8}.
\end{align*}
Computation of $\mathcal{A}_{i,j}^{[\textsf{n}_{i,j}]}$ is performed with the help of
 \texttt{NIntegrate} with  a working precision of $18$.
The relations \eqref{C34Rels} and \eqref{C34Rels2} are satisfied with an accuracy of $10^{-15}$.

The following not normalized period matrices are obtained
\begin{align*}
&\omega \approx \begin{pmatrix}
-0.646716 \imath   & 1.367235 \imath & -1.406214 \imath \\
 0.557691  \imath & 0.662524 \imath & 0.237700 \imath \\
-0.425220 \imath & -0.085658 \imath  & 0.761823 \imath 
\end{pmatrix},\\
&\omega' \approx \begin{pmatrix}
1.114221 + 0.360259 \imath  & 
 -0.838244 + 0.360259 \imath  & 
 0.830310 - 0.703107 \imath \\
 -0.888801 + 0.610108 \imath & 
 -0.725076 + 0.610108 \imath & 
 -0.483530 + 0.118850 \imath  \\
 0.212490 - 0.255439 \imath &
 0.017209 - 0.255439 \imath&
 -0.244951 + 0.380911 \imath
\end{pmatrix},
\end{align*}
and the normalized period matrix from the Siegel upper half-space
\begin{align*}
&\tau \approx \begin{pmatrix}
0.5 + 1.204054 \imath & 0.5 + 0.179707 \imath & 0.413339 \imath \\
0.5 + 0.179707 \imath & 0.5 + 0.879769 \imath & 0.176635 \imath \\
0.413339 \imath & 0.176635 \imath & 0.5 + 0.572103 \imath
\end{pmatrix}.
\end{align*}
Second kind period matrices are
\begin{align*}
&\eta \approx \begin{pmatrix}
 - 0.541959 \imath & - 0.385425 \imath & -0.722057 \imath \\ 
   1.52536 \imath &  - 0.88414 \imath &  0.484784 \imath \\ 
   0.975636 \imath & - 1.01088 \imath  &  -2.65892 \imath
\end{pmatrix},\\
&\eta' \approx \begin{pmatrix}
-1.357307 - 0.463692 \imath & 
 2.945439  - 0.463692 \imath &
- 0.354124 -  0.361028 \imath \\
 2.292766 + 0.320609 \imath & 
 5.356432 + 0.320609 \imath & 
 2.131611 + 0.242392 \imath \\
 -5.584050 - 0.017623 \imath &
  4.038588 - 0.017623 \imath & 
  6.689080 - 1.329459 \imath
\end{pmatrix},
\end{align*}
and the symmetric matrix $\varkappa$ is
\begin{align*}
&\varkappa \approx \begin{pmatrix}
0.180731 & -0.994032 & -0.304044 \\
-0.994032 & 0.540116 & -1.367017 \\
-0.304044 &  -1.367017  & -3.624898
\end{pmatrix}.
\end{align*}
 The symmetric property of $\tau$ is satisfied with an accuracy of $10^{-15}$, and
 of $\varkappa$ with an accuracy of $10^{-13}$.

\section{Computation of $\wp$-functions on a trigonal curve}\label{s:TrigWP}
Similar to the hyperelliptic case,
 the Abel image $\mathcal{A}(D)$ of a given divisor $D$ is computed directly
 by the formula \eqref{AMapD}, where the Abel image $\mathcal{A}(P)$ of a point  $P$
 is computed by \eqref{AbelM} with
 the standard not normalized holomorphic differentials \eqref{Dif1TC}.

$\wp$-Functions are calculated on non-special divisors by means of \eqref{WPdef}. 
Non-special divisors $D$ are composed as positive divisors of degree $n\geqslant g$ 
with no three points in involution on a trigonal curve of genus $g$. 

\subsection{Vector of Riemann constants}
The characteristic $[K]$ of the vector of Riemann constants $K$
is required for computation of $\wp$-functions by means of \eqref{WPdef}.
This characteristic is half-integer, due to the relation $2 K \sim \bar{\mathcal{A}}(C)$, see \cite[Eq.\,(2.4.20)]{Dub1981},
which means that $2K$ is congruent to the Abel image of the canonical divisor $C$
on~$\mathcal{C}$, and
the latter is congruent to zero. Recall, that $\bar{\mathcal{A}}$ denotes the Abel map  
with respect to normalized differentials. Moreover,
\begin{theo}\label{T:Kval}
The theta function with  characteristic $[K]$, as a function of not normalized 
coordinates $u$, is characterized by the maximal order $\mathfrak{d}$ of vanishing at $u=0$, that is
$$ \forall \mathfrak{i}<\mathfrak{d}\quad 
\frac{\partial^{\mathfrak{i}} \theta[K](\omega^{-1} u) }{\partial u_1^{\mathfrak{i}}}\Big|_{u=0} = 0,\qquad\quad
 \frac{\partial^{\mathfrak{d}} \theta[K](\omega^{-1} u) }{\partial u_1^{\mathfrak{d}}}\Big|_{u=0} \neq 0, $$
where  $\mathfrak{d}=(3 \mFr + 2) \mFr$ on a $(3,3\mFr+1)$-curve, and
$\mathfrak{d}=(3 \mFr + 1) (\mFr + 1)$ on a $(3,3\mFr+2)$-curve; 
the order $\mathfrak{d}$ is computed with respect to the Sato weight.
\end{theo}
\begin{proof}
By \eqref{SigmaThetaRel} $\theta [K]$
is connected to $\sigma$-function, which means that the both functions 
behave similarly at the origin $u=0$ of $\Jac(\mathcal{C})$.
From \cite{bel99} we know the Sato weight of  $\sigma$-function:
$$ \mathrm{wgt} \; \sigma = -\tfrac{1}{24} (n^2-1)(s^2-1).$$
The negative Sato weight shows the  order of vanishing at $u=0$,
that is $\mathfrak{d} = - \mathrm{wgt} \; \sigma$.
By direct computations, one can find, that $\wgt \sigma = (3 \mFr + 2) \mFr$
on a $(3,3\mFr+1)$-curve, and $\wgt \sigma = (3 \mFr + 1) (\mFr + 1)$
on a $(3,3\mFr+2)$-curve.
 \end{proof}
 
\begin{cor} 
If the weighted order $\mathfrak{p}$ of a derivative of $\theta[K](\omega^{-1} u)$
with respect to $u$
is less than $\mathfrak{d}$, then the derivative vanishes.
\end{cor}
\begin{proof}
The weighted order of a derivative with respect to $u = (u_{\mathfrak{w}_1},\dots,
u_{\mathfrak{w}_g})^t$ is defined as follows
$$\mathfrak{p} = \ord \frac{\partial^{p_{\mathfrak{w}_1} + \cdots + p_{\mathfrak{w}_g}} }
{\partial u_{\mathfrak{w}_1}^{p_{\mathfrak{w}_1}} \cdots \partial u_{\mathfrak{w}_g}^{p_{\mathfrak{w}_g}} } 
= \sum_{i=1}^g \mathfrak{w}_i p_{\mathfrak{w}_i }. $$
All derivatives of $\sigma$-function such that $\mathfrak{p} < \mathfrak{d}$  vanish at $u=0$,
due to the Riemann vanishing theorem.   
 \end{proof} 
 
From the relation $[K] = [ \frac{1}{2} \bar{\mathcal{A}}(C)]$
 the exact location of the base-point for computation can be found. 
Though it is known that the base-point is located at infinity, the computational base-point
is a matter of investigation.
Unlike the hyperelliptic case, all paths to points in calculating Abel images
are required to start at the same computational base-point 
 on a fixed sheet.

\subsection{Paths to points}
Let $P_i =(x_i,y_i)$,  $y_i=y_{a}(x_i)$. From  $a$ the Sheet\;$\textsf{n}$
where $P_i$ is located is identified. Let $Q_i = (e_i,y_{a}(e_i))$, where $e_i$
is the $x$-coordinate of a branch point $B_i$,
 be in the vicinity to $P_i$ on  Sheet\;$\textsf{n}$.
 A path to $P_i$  starts at $-\infty-\imath \epsilon$ on the fixed sheet, 
and goes  to $Q_i$ on Sheet\;$\textsf{n}$. 
Then the segment $[Q_i,P_i]$ on Sheet\;$\textsf{n}$ is added to this path. 
Along the path to $P_i$ the Abel image $\mathcal{A}(P_i)$ is computed.

Below, computation of $\wp$-functions is illustrated by examples.

\subsection{Example 3a}
On the curve \eqref{C34Eq}, with the homology basis chosen as shown on fig.\,\ref{f:C34Cycles},
the vector of Riemann constants has the characteristic
\begin{equation}\label{KChar}
 [K] = \small \begin{pmatrix}
1 & 0 & 1 \\
0 & 1 & 1 \end{pmatrix}, 
\end{equation}
and $\theta [K]$ vanishes to the order $5$ at $u=0$. 

Recall that all branch points are located on Sheet\;\textsf{c}.
Let  paths to all these points start at $-\infty-\imath \epsilon$ on Sheet\;\textsf{c}, 
and reach branch points in the shortest way.
Then
\begin{multline}\label{ACanonDiv}
\omega \bar{\mathcal{A}}(C) = \mathcal{A}^{[3]}_{0-,1} 
+ \big(\mathcal{A}^{[3]}_{0-,1} + \mathcal{A}^{[3]}_{1,2} \big)
+ \big(\mathcal{A}^{[3]}_{0-,1} + \mathcal{A}^{[3]}_{1,2} + \mathcal{A}^{[1]}_{2,3} \big) \\
+ \big(\mathcal{A}^{[3]}_{0-,1} + \mathcal{A}^{[3]}_{1,2} + \mathcal{A}^{[1\text{-}2]}_{2,4} \big)
+ \big(\mathcal{A}^{[3]}_{0-,1} + \mathcal{A}^{[3]}_{1,2} + \mathcal{A}^{[1\text{-}2]}_{2,4} 
+ \mathcal{A}^{[3\text{-}1]}_{4,6} + \mathcal{A}^{[1\text{-}2]}_{6,5} \big) \\
+ \big(\mathcal{A}^{[3]}_{0-,1} + \mathcal{A}^{[3]}_{1,2} + \mathcal{A}^{[1\text{-}2]}_{2,4} 
+ \mathcal{A}^{[3\text{-}1]}_{4,6} \big) \\
+ \big(\mathcal{A}^{[3]}_{0-,1} + \mathcal{A}^{[3]}_{1,2} + \mathcal{A}^{[1\text{-}2]}_{2,4} 
+ \mathcal{A}^{[3\text{-}1]}_{4,6} + \mathcal{A}^{[2\text{-}3]}_{6,7} \big) \\
+ \big(\mathcal{A}^{[3]}_{0-,1} + \mathcal{A}^{[3]}_{1,2} + \mathcal{A}^{[1\text{-}2]}_{2,4} 
+ \mathcal{A}^{[3\text{-}1]}_{4,6} + \mathcal{A}^{[2\text{-}3]}_{6,8} \big) 
\approx \begin{pmatrix} 
 3.5 + 1.976806 \imath \\
-1.5 + 2.115880 \imath \\
 1.5 + 1.338712 \imath
  \end{pmatrix},
\end{multline}
and in terms of columns of the normalized periods $(1_3, \tau)$:
$$ \big(\bar{\mathcal{A}}(C) \big)_j = \tau_{j,1} + 2 \tau_{j,2} + 
 \tau_{j,3} + 2 \delta_{j,1} - 3 \delta_{j,2} + \delta_{j,3},$$
where $\delta_{i,j}$ denotes the Kronecker delta. Evidently, 
$\frac{1}{2} \bar{\mathcal{A}}(C)$ obtained from \eqref{ACanonDiv}
gives the correct value of the vector of Riemann constants, since 
its characteristic coincides with  \eqref{KChar}.
Therefore, we fix the computational base-point at $-\infty-\imath \epsilon$ on Sheet\;\textsf{c}.

\medskip
Let a given divisor be $D = \sum_{i=1}^3 P_i$,
\begin{equation}\label{Ex3aPs}
\begin{split}
P_1  &= \big(e_3 - 0.5-0.5 \imath, y_1(e_3 - 0.5-0.5 \imath)\big) \\
&\approx (-1.679223 + 0.434455 \imath, 3.431176 - 0.582699 \imath), \\
P_2 &= \big(e_5 + 0.5-1.5 \imath, y_1(e_5 + 0.5-1.5 \imath)\big)  \\
&\approx (0.068268 + 0.702564 \imath, 3.555003 + 0.889027 \imath),\\
P_3 &= \big(e_6 +1, y_1(e_6 +1)\big) \\
&\approx (1.499118 - 1.575269 \imath, -4.191263 + 1.058317 \imath).
\end{split}
\end{equation}
The points $P_1$ and $P_3$ are located on Sheet\;\textsf{a}, and $P_2$ on Sheet\;\textsf{b}.
\begin{figure}[h]
\parbox[b]{0.38\textwidth}{\includegraphics[width=0.35\textwidth]{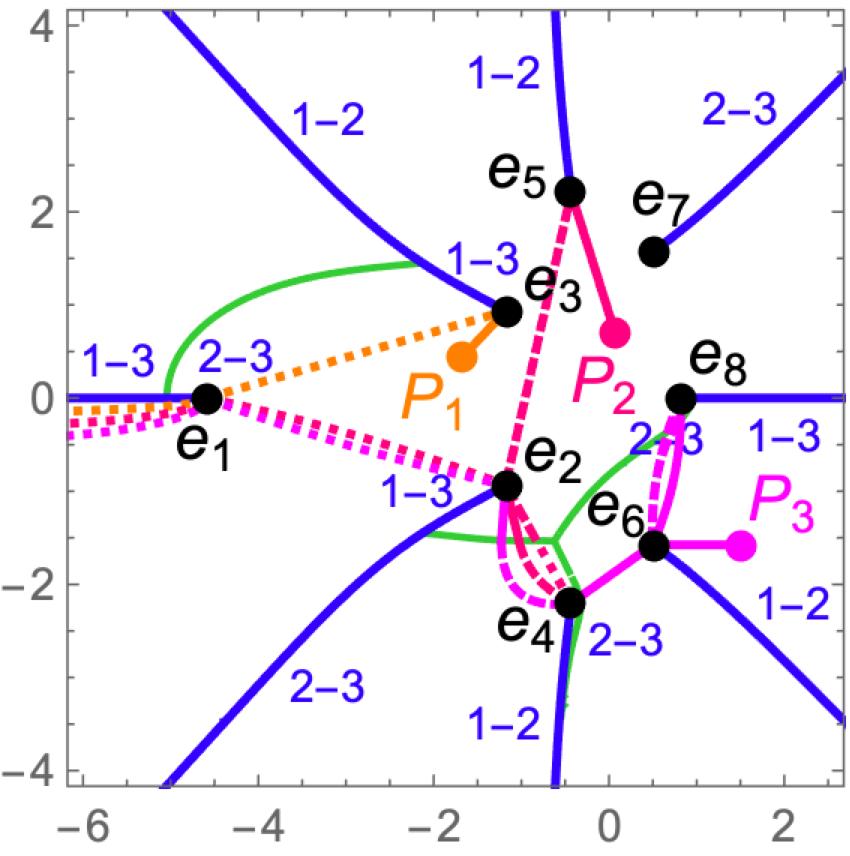}}$\quad$
\parbox[b]{0.38\textwidth}{\includegraphics[width=0.35\textwidth]{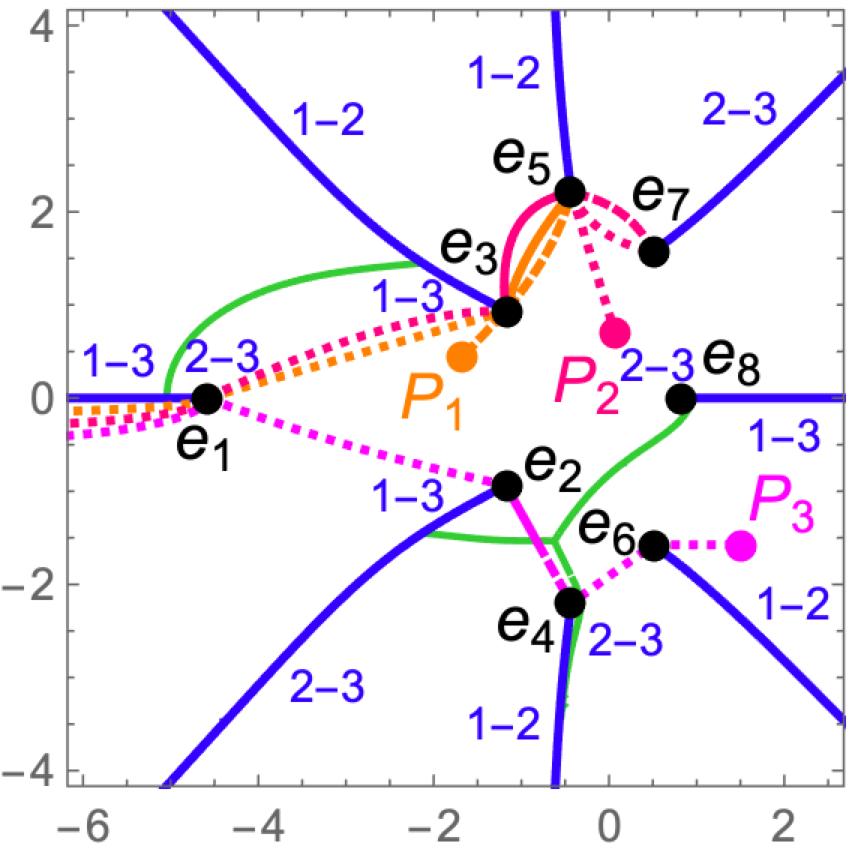}}$\quad$
\parbox[b]{0.15\textwidth}{\includegraphics[width=0.1\textwidth]{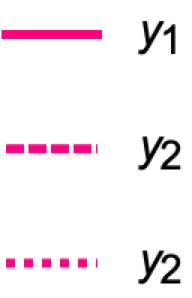}\\
 $\quad$ \\  $\quad$ \\  $\quad$ \\  $\quad$ \\  $\quad$}\\
\parbox[b]{0.4\textwidth}{\centering (a) Example 3a}$\ \ $
\parbox[b]{0.4\textwidth}{\centering (b) Example 3b}$\qquad\quad\quad\quad\quad$
\caption{Paths to points of divisor $D$.}\label{f:C34Ex3}
\end{figure}
A path to each point starts at $-\infty - \imath \epsilon$ on Sheet\;\textsf{c}. 
Let $P_1$ on Sheet\;\textsf{a} be reached through the cut  $(B_2, B_3)$,
then $P_2$ on Sheet\;\textsf{b} through the cut  $(B_4,B_5)$,
and $P_3$ on Sheet\;\textsf{a} through the cut $(B_8,\infty)\cup(\infty, B_1)$, see fig.\,\ref{f:C34Ex3}(a). 
Actually,
\begin{align*}
&\mathcal{A}(P_1) = \mathcal{A}_{0-,1}^{[3]} + \mathcal{A}_{1,3}^{[3]}  + \int_{e_3}^{P_1} \rmd u^{[1]},\\
&\mathcal{A}(P_2) = \mathcal{A}_{0-,1}^{[3]} + \mathcal{A}_{1,2}^{[3]} + \mathcal{A}_{2,4}^{[1\text{-}2]} 
+ \mathcal{A}_{4,2}^{[3\text{-}2]} +  \mathcal{A}_{2,5}^{[2]} + \int_{e_5}^{P_2} \rmd u^{[1]},\\
&\mathcal{A}(P_3) = \mathcal{A}_{0-,1}^{[3]} + \mathcal{A}_{1,2}^{[3]}  + \mathcal{A}_{2,4}^{[1\text{-}2]} 
+ \mathcal{A}_{4,6}^{[3\text{-}1]} + \int_{e_6}^{P_3} \rmd u^{[1]},
\end{align*}
where $\rmd u^{[a]}$ means that  $y = y_a(x)$.
The Abel image of $D$ is
\begin{equation}\label{AMap3a1}
 u(D) = \sum_{i=1}^3 \mathcal{A}(P_i) \approx \begin{pmatrix}
 -0.270333 - 1.612257 \imath \\  -1.116879 + 0.562199 \imath \\
  0.258194 + 0.268653 \imath
\end{pmatrix}. 
\end{equation}

By means of \eqref{WPdef}, with $\theta$-function approximated by
a partial sum of \eqref{ThetaDef},  $n_i \leqslant 5$, we obtain 
\begin{align}\label{WPvalsC34Ex3a}
&\wp_{1,1}\big(u(D)\big) \approx 0.059654 + 1.020925 \imath,& \notag \\
&\wp_{1,2}\big(u(D)\big) \approx -0.793416 + 0.889005 \imath,& \notag \\
&\wp_{1,5}\big(u(D)\big) \approx 0.885372 - 3.089764 \imath,& \notag \\
&\wp_{2,2}\big(u(D)\big) \approx -0.269700 + 1.472739 \imath,& \\
&\wp_{2,5}\big(u(D)\big) \approx -3.501466 + 10.538856 \imath, & \notag \\
&\wp_{1,1,1}\big(u(D)\big) \approx -2.156576 + 3.543516 \imath,  \notag\\
&\wp_{1,1,2}\big(u(D)\big) \approx -3.595029 + 2.840859 \imath, \notag \\
&\wp_{1,1,5}\big(u(D)\big) \approx 5.656516 - 0.559812 \imath.  \notag
\end{align}

Alternatively, a part of the path $\bar{\gamma}_{3\infty}$ around infinity 
can be used to reach each point on its sheet avoiding cuts.
Denote
\begin{align*}
&\mathcal{A}(\tilde{\gamma}_{\text{L}}^{\textsf{c}}) 
= \mathcal{A}_{1,2}^{[3]}  + \mathcal{A}_{2,4}^{[1\text{-}2]} 
+ \mathcal{A}_{4,6}^{[3\text{-}1]} + \mathcal{A}_{6,8}^{[2\text{-}3]}, \\
&\mathcal{A}(\tilde{\gamma}_{\text{L}}^{\textsf{a}}) 
= \mathcal{A}^{[1]}_{1,2} + \mathcal{A}^{[3\text{-}1]}_{2,4} 
+ \mathcal{A}^{[1\text{-}2]}_{4,6} + \mathcal{A}^{[1\text{-}2]}_{6,8},\\
&\mathcal{A}(\tilde{\gamma}_{\text{L}}^{\textsf{b}}) 
= \mathcal{A}^{[2]}_{1,2}  + \mathcal{A}^{[2\text{-}3]}_{2,4} 
+ \mathcal{A}^{[2\text{-}3]}_{4,6} + \mathcal{A}^{[3\text{-}1]}_{6,8}, \\
&\mathcal{A}(\tilde{\gamma}_{\text{U}}^{\textsf{a}}) 
= \mathcal{A}^{[2]}_{8,7}  + \mathcal{A}^{[3]}_{7,5}  
+ \mathcal{A}^{[3]}_{5,3}  + \mathcal{A}^{[1]}_{3,1},\\
&\mathcal{A}(\tilde{\gamma}_{\text{U}}^{\textsf{c}}) 
= \mathcal{A}^{[1]}_{8,7} + \mathcal{A}^{[1]}_{7,5} 
+ \mathcal{A}^{[2]}_{5,3} + \mathcal{A}^{[2]}_{3,1},\\
&\mathcal{A}(\tilde{\gamma}_{\text{U}}^{\textsf{b}}) 
= \mathcal{A}^{[3]}_{8,7} + \mathcal{A}^{[2]}_{7,5} 
+ \mathcal{A}^{[1]}_{5,3} + \mathcal{A}^{[3]}_{1,3} ,
\end{align*}
cf.\,\eqref{C34Rels} and \eqref{C34Rels2}, then
\begin{align*}
&\mathcal{A}(\tilde{\gamma}_\infty^{\textsf{c-a}}) = 
\mathcal{A}(\tilde{\gamma}_{\text{L}}^{\textsf{c}}) 
+ \mathcal{A}(\tilde{\gamma}_{\text{U}}^{\textsf{a}}),\\
&\mathcal{A}(\tilde{\gamma}_\infty^{\textsf{a-c}}) = 
\mathcal{A}(\tilde{\gamma}_{\text{L}}^{\textsf{a}}) 
+ \mathcal{A}(\tilde{\gamma}_{\text{U}}^{\textsf{c}}),\\
&\mathcal{A}(\tilde{\gamma}_\infty^{\textsf{b-b}}) = 
\mathcal{A}(\tilde{\gamma}_{\text{L}}^{\textsf{b}}) 
+ \mathcal{A}(\tilde{\gamma}_{\text{U}}^{\textsf{b}}).
\end{align*}
$\mathcal{A}(\tilde{\gamma}_{\text{L}}^{\textsf{c}})$ can be used to reach the vicinity of $B_8$ on Sheet\;\textsf{a},
and $\mathcal{A}(\tilde{\gamma}_\infty^{\textsf{c-a}})$ to reach the vicinity of $B_1$ on Sheet\;\textsf{a}.
By $\mathcal{A}(\tilde{\gamma}_\infty^{\textsf{c-a}}) 
+\mathcal{A}(\tilde{\gamma}_\infty^{\textsf{a-c}})$ the vicinity of $B_1$ on Sheet\;\textsf{b} can be reached,
and by $\mathcal{A}(\tilde{\gamma}_\infty^{\textsf{c-a}}) +\mathcal{A}(\tilde{\gamma}_\infty^{\textsf{a-c}})
+ \mathcal{A}(\tilde{\gamma}_{\text{L}}^{\textsf{b}})$ the vicinity of $B_8$ on Sheet\;\textsf{b}.
Then
\begin{align*}
&\mathcal{A}(P_1) = \mathcal{A}_{0-,1}^{[3]} + \mathcal{A}(\tilde{\gamma}_\infty^{\textsf{c-a}}) 
+ \mathcal{A}_{1,3}^{[1]}  + \int_{e_3}^{P_1} \rmd u^{[1]},\\
&\mathcal{A}(P_2) = \mathcal{A}_{0-,1}^{[3]} + \mathcal{A}(\tilde{\gamma}_\infty^{\textsf{c-a}}) 
+\mathcal{A}(\tilde{\gamma}_\infty^{\textsf{a-c}}) 
+ \mathcal{A}_{1,2}^{[2]} + \mathcal{A}_{2,4}^{[2\text{-}3]} 
+ \mathcal{A}_{4,6}^{[2\text{-}3]}  +  \mathcal{A}_{6,5}^{[3\text{-}1]} + \int_{e_5}^{P_2} \rmd u^{[1]},\\
&\mathcal{A}(P_3) = \mathcal{A}_{0-,1}^{[3]} + \mathcal{A}(\tilde{\gamma}_\infty^{\textsf{c-a}}) 
+ \mathcal{A}_{1,2}^{[1]}  + \mathcal{A}_{2,4}^{[3\text{-}1]} 
+ \mathcal{A}_{4,6}^{[1\text{-}2]} + \int_{e_6}^{P_3} \rmd u^{[1]},
\end{align*}
which produce an Abel image congruent to \eqref{AMap3a1}, namely
\begin{equation}\label{AMap3a2}
u(D) = \sum_{i=1}^3 \mathcal{A}(P_i) \approx \begin{pmatrix}
-0.546310 + 0.440673 \imath \\  0.496998 - 0.233192 \imath \\
 0.028495 - 0.067950 \imath
 \end{pmatrix}.
\end{equation}
Values of $\wp$-functions on \eqref{AMap3a2} and \eqref{AMap3a1} coincide with
an accuracy of $10^{-13}$.

A solution of the Jacobi inversion problem on a $(3,4)$-curve \eqref{C34Eq} is given by the system
\begin{subequations}\label{C34JIP}
\begin{align}
&\mathcal{R}_6(x,y;u) \equiv x^2 - y \wp_{1,1}(u) - x \wp_{1,2}(u) - \wp_{1,5}(u)=0,\\
&\mathcal{R}_7(x,y;u) \equiv 2 x y + y \big(\wp_{1,1,1}(u) - \wp_{1,2}(u)\big) \\
&\phantom{\mathcal{R}_7(x,y;u)}\quad + x \big(\wp_{1,1,2}(u) - \wp_{2,2}(u)\big) 
+ \big(\wp_{1,1,5}(u) - \wp_{2,5}(u)\big)=0, \notag
\end{align}
\end{subequations}
whose divisor of zeros is a degree $3$ positive divisor $D_3$ such that $u = \mathcal{A}(D_3)$.

On the other hand, the two polynomial functions $\mathcal{R}_{6}$ and $\mathcal{R}_{7}$
on a $(3,4)$-curve \eqref{C34Eq}
can be constructed directly from coordinates of points of a degree $3$ divisor 
 $D_3 = \sum_{i=1}^{3} (x_i,y_i)$. Namely,
\begin{align}\label{R67Coord}
&\mathcal{R}_6(x,y;D_3) = \frac{\small 
\begin{vmatrix} x^2 & y & x & 1 \\
x_1^2 & y_1 & x_1 & 1 \\
x_2^2 & y_2 & x_2 & 1 \\
x_3^2 & y_3 & x_3 & 1 \end{vmatrix}}
{\small \begin{vmatrix} 
y_1 & x_1 & 1  \\
y_2 & x_2 & 1  \\
y_3 & x_3 & 1  \end{vmatrix}},&
&\mathcal{R}_7(x,y;D_3) = 2 \frac{\small \begin{vmatrix} x y & y & x & 1  \\
x_1 y_1 & y_1 & x_1 & 1 \\
x_2 y_2 & y_2 & x_2 & 1 \\
x_3 y_3 & y_3 & x_3 & 1 \end{vmatrix}}
{\small \begin{vmatrix} 
y_1 & x_1 & 1  \\
y_2 & x_2 & 1  \\
y_3 & x_3 & 1  \end{vmatrix}}.
\end{align}
On the given divisor $D$ defined by \eqref{Ex3aPs} we obtain 
\begin{subequations}
\begin{align*}
\mathcal{R}_{6}\big(x,y;D\big) &= x^2 -
(0.059654 + 1.020925  \imath) y \\
&\quad + (0.793416 - 0.889005  \imath) x - 0.885372 + 3.089764  \imath, \notag \\ 
\mathcal{R}_{7}\big(x,y;D\big) &= 2 x y - (1.363160 - 2.654511 \imath) y \\
&\quad - (3.325329 - 1.368120 \imath) x + 9.157983 - 11.098669 \imath. \notag
\end{align*}
\end{subequations}
Coefficients of these two functions coincide with those expressed in terms of 
 $\wp$-functions \eqref{WPvalsC34Ex3a}  within an accuracy of $10^{-13}$.

\subsection{Example 3b}
Let $D$ be slightly modified by choosing points with the same $x$-coordinates, 
but located on other sheets, namely
\begin{align}\label{Ex3bPs}
\begin{split}
P_1  &= \big(e_3 - 0.5-0.5 \imath, y_2(e_3 - 0.5-0.5 \imath)\big) \\
&\approx (-1.679223 + 0.434455 \imath, -3.855009 + 0.446784 \imath), \\
P_2 &= \big(e_5 + 0.5-1.5 \imath, y_3(e_5 + 0.5-1.5 \imath)\big)  \\
&\approx (0.068268 + 0.702564 \imath, -0.916997 - 0.935049 \imath),\\
P_3 &= \big(e_6 +1, y_3(e_6 +1)\big) \\
&\approx (1.499118 - 1.575269 \imath, 5.701625 - 3.318255  \imath).
\end{split}
\end{align}
The points $P_1$ and $P_3$ are located on Sheet\;\textsf{b}, and $P_2$ on Sheet\;\textsf{a}.
Let $P_1$ and $P_3$ on Sheet\;\textsf{b} be reached through the cut  $(B_4,B_5)$,
and $P_2$ on Sheet\;\textsf{a} through the cut  $(B_6,B_7)$, see fig.\,\ref{f:C34Ex3}(b). 
Actually,
\begin{align*}
&\mathcal{A}(P_1) = \mathcal{A}_{0-,1}^{[3]} + \mathcal{A}_{1,3}^{[3]} + \mathcal{A}_{3,5}^{[1]}
+  \mathcal{A}_{5,3}^{[2]} + \int_{e_3}^{P_1} \rmd u^{[2]},\\
&\mathcal{A}(P_2) = \mathcal{A}_{0-,1}^{[3]} + \mathcal{A}_{1,3}^{[3]} + \mathcal{A}_{3,5}^{[1]} + \mathcal{A}_{5,7}^{[2]}
+ \mathcal{A}_{7,5}^{[3]} + \int_{e_5}^{P_2} \rmd u^{[3]},\\
&\mathcal{A}(P_3) = \mathcal{A}_{0-,1}^{[3]} + \mathcal{A}_{1,2}^{[3]}  + \mathcal{A}_{2,4}^{[1\text{-}2]} 
+ \mathcal{A}_{4,6}^{[2\text{-}3]} + \int_{e_6}^{P_3} \rmd u^{[3]}.
\end{align*}
The Abel image of $D$ is
\begin{equation}\label{AMap3b}
 u(D) = \sum_{i=1}^3 \mathcal{A}(P_i) \approx \begin{pmatrix}
 -0.421105 - 2.303962 \imath \\ -1.319230 - 1.997581 \imath \\ 
 -0.176345 + 0.125109 \imath
\end{pmatrix},
\end{equation}
and $\wp$-functions acquire the following values
\begin{gather}\label{WPvalsC34Ex3b}
\begin{aligned}
&\wp_{1,1}\big(u(D)\big) \approx -0.497171- 1.306218 \imath,& \\
&\wp_{1,2}\big(u(D)\big) \approx 0.485105 + 2.618402 \imath,& \\
&\wp_{1,5}\big(u(D)\big) \approx 2.083016 - 2.086324 \imath,& \\
&\wp_{2,2}\big(u(D)\big) \approx -2.356414 + 10.869587 \imath,& \\
&\wp_{2,5}\big(u(D)\big) \approx 15.590831 + 2.902800 \imath, & \\
&\wp_{1,1,1}\big(u(D)\big) \approx 1.678988 + 8.731706 \imath, \\
 &\wp_{1,1,2}\big(u(D)\big) \approx -4.377331 - 0.119524 \imath, \\
 &\wp_{1,1,5}\big(u(D)\big) \approx 2.198126 + 13.211222 \imath.
\end{aligned}
\end{gather}
With these values the two functions $\mathcal{R}_{6}$, $\mathcal{R}_{7}$ in \eqref{C34JIP}
are constructed.

On the other hand, the two polynomial functions $\mathcal{R}_{6}$, $\mathcal{R}_{7}$
are obtained by \eqref{R67Coord}
from coordinates of the given divisor $D$ defined by \eqref{Ex3bPs}, namely
\begin{subequations}
\begin{align*}
\mathcal{R}_{6}\big(x,y;D\big) &= x^2 +
(0.497171 + 1.306218\imath) y \\
&\quad - (0.485105 + 2.618402 \imath) x - 2.083016 + 2.086324  \imath, \notag \\ 
\mathcal{R}_{7}\big(x,y;D\big) &= 2 x y + (1.193883 + 6.113304 \imath) y \\
&\quad - (2.020917 + 10.989112 \imath) x -13.392705 + 10.308422 \imath. \notag
\end{align*}
\end{subequations}
 Coefficients of these two functions coincide with those computed from
 $\wp$-functions \eqref{WPvalsC34Ex3b}  within an accuracy of $10^{-9}$.

\section{Acknowledgments}
The present paper was inspired by S.\,Matsutani, who expressed an active interest into
 analytical computation of  $\wp$-functions in Mathematica.


\end{document}